\documentclass[journal,twoside,web]{ieeecolor}

\usepackage{generic}
\usepackage{epstopdf}

\usepackage{cite}

\usepackage{amssymb,amsfonts}
\usepackage[fleqn]{amsmath}
\usepackage{graphicx}
\usepackage{textcomp}
\usepackage{subcaption}
\usepackage[ruled,vlined]{algorithm2e}
\usepackage{float}

\usepackage{xspace}
\usepackage{soul}
\usepackage{mathtools}
\usepackage{makecell}
\usepackage{todonotes}
\usepackage{nicefrac}

\newcommand{\ra}{\rightarrow}

\newcommand{\la}{\leftarrow}

\newcommand{\ie}{\unskip, i.\,e.,\xspace}
\newcommand{\eg}{\unskip, e.\,g.,\xspace}

\newcommand{\sut}{\text{s.\,t.\,}}

\newcommand{\nrm}[1]{\left\lVert#1\right\rVert}
\newcommand{\diag}[1]{\ensuremath{\text{diag}}\left\{#1\right\}}
\newcommand{\abs}[1]{\left\lvert#1\right\rvert}

\newcommand{\E}[1]{\mathbb E\left[#1\right]}

\newcommand{\N}{\ensuremath{\mathbb{N}}}

\newcommand{\R}{\ensuremath{\mathbb{R}}}

\newcommand{\X}{\ensuremath{\mathbb{X}}}

\newcommand{\U}{\ensuremath{\mathbb{U}}}
\newcommand*\diff{\mathop{}\!\mathrm{d}}
\newcommand{\eps}{\ensuremath{\varepsilon}}

\newcommand{\spc}{\ensuremath{\,\,}}

\newcommand{\ball}{\ensuremath{\mathcal B}}

\DeclareMathOperator*{\argmin}{arg\,min}

\newcommand{\red}[1]{\textcolor{red}{#1}}

\definecolor{dgreen}{rgb}{0.0, 0.5, 0.0}
\newcommand{\green}[1]{\textcolor{dgreen}{#1}}

\usepackage{float}
\usepackage{wrapfig}
\usepackage{xcolor}

\newtheorem{dfn}{Definition}
\newtheorem{asm}{Assumption}

\newtheorem{lem}{Lemma}
\newtheorem{thm}{Theorem}

\newtheorem{rem}{Remark}
\newcommand{\W}{\ensuremath{\mathbb{W}}}
\newcommand{\B}{\ensuremath{\mathcal{B}}}
\newcommand{\K}{\ensuremath{\mathcal{K}}\xspace}				% Capital kappa
\newcommand{\Kinf}{\ensuremath{\mathcal{K}_{\infty}}\xspace}		% Kappa-infinity
\newcommand{\alphalow}{\ensuremath{\alpha_{\text{low}}}}
\newcommand{\alphaup}{\ensuremath{\alpha_{\text{up}}}}
\newcommand{\hatalphalow}{\ensuremath{\hat\alpha_{\text{low}}}}
\newcommand{\hatalphaup}{\ensuremath{\hat\alpha_{\text{up}}}}

\newcommand{\inv}{\ensuremath{^{-1}}}

\makeatletter
\newcommand{\subalign}[1]{%
	\vcenter{%
		\Let@ \restore@math@cr \default@tag
		\baselineskip\fontdimen10 \scriptfont\tw@
		\advance\baselineskip\fontdimen12 \scriptfont\tw@
		\lineskip\thr@@\fontdimen8 \scriptfont\thr@@
		\lineskiplimit\lineskip
		\ialign{\hfil$\m@th\scriptstyle##$&$\m@th\scriptstyle{}##$\crcr
			#1\crcr
		}%
	}
}

\makeatletter
\def\oversortoftilde#1{\mathop{\vbox{\m@th\ialign{##\crcr\noalign{\kern3\p@}%
				\sortoftildefill\crcr\noalign{\kern3\p@\nointerlineskip}%
				$\hfil\displaystyle{#1}\hfil$\crcr}}}\limits}

\def\sortoftildefill{$\m@th \setbox\z@\hbox{$\braceld$}%
	\braceld\leaders\vrule \@height\ht\z@ \@depth\z@\hfill\braceru$}

\makeatother

\newcommand{\lip}[1]{\text{Lip}_{#1}}

\newcommand{\wo}{{w_{\text{prev}}}}
\newcommand{\xn}{{x_{\text{next}}}}
\newcommand{\sublyapunov}{\hat{\alpha}_{\text{low}}^{-1}(\alpha_\text{up}(\nrm{X_0}))}

\newcommand{\Ja}{\ensuremath{\mathbb{J}_a}}
\newcommand{\Jc}{\ensuremath{\mathbb{J}_c}}

\def\BibTeX{{\rm B\kern-.05em{\sc i\kern-.025em b}\kern-.08em
    T\kern-.1667em\lower.7ex\hbox{E}\kern-.125emX}}
\begin{document}

\title{A framework for online, stabilizing reinforcement learning}
\author{Grigory Yaremenko, Georgy Malaniya, Pavel Osinenko
\thanks{The authors are with the Skolkovo Institute of Science and Technology, emails: \texttt{\{Grigory.Yaremenko, p.osinenko\}@skoltech.ru}.}
\thanks{The authors are thankful to Ksenia Makarova for supporting the experimental studies of this work.}
\thanks{Pavel Osinenko is the corresponding author.}
}

\maketitle

%%%%%%%%%%%%%%%%%%%%%%%%%%%%%%%%%%%%%%%%%%%%%%%%%%%%%%%%%%%%%%%%%%%%%%%%%%%%%%%%%%%%%%%%%%%%
%%%%%%%%%%%%%%%%%%%%%%%%%%%%%%%%%%%%%%%%%%%%%%%%%%%%%%%%%%%%%%%%%%%%%%%%%%%%%%%%%%%%%%%%%%%%

\begin{abstract}
Online reinforcement learning is concerned with training an agent on-the-fly via dynamic interaction with the environment.
Here, due to the specifics of the application, it is not generally possible to perform long pre-training, as it is commonly done in off-line, model-free approaches, which are akin to dynamic programming.
Such applications may be found more frequently in industry, rather than in pure digital fields, such as cloud services, video games, database management, etc., where reinforcement learning has been demonstrating success.
Online reinforcement learning, in contrast, is more akin to classical control, which utilizes some model knowledge about the environment.
Stability of the closed-loop (agent plus the environment) is a major challenge for such online approaches.
In this paper, we tackle this problem by a special fusion of online reinforcement learning with elements of classical control, namely, based on the Lyapunov theory of stability.
The idea is to start the agent at once, without pre-training, and learn approximately optimal policy under specially designed constraints, which guarantee stability.
The resulting approach was tested in an extensive experimental study with a mobile robot.
A nominal parking controller was used as a baseline.
It was observed that the suggested agent could always successfully park the robot, while significantly improving the cost.
While many approaches may be exploited for mobile robot control, we suggest that the experiments showed the promising potential of online reinforcement learning agents based on Lyapunov-like constraints.
The presented methodology may be utilized in safety-critical, industrial applications where stability is necessary.
\end{abstract}


%%%%%%%%%%%%%%%%%%%%%%%%%%%%%%%%%%%%%%%%%%%%%%%%%%%%%%%%%%%%%%%%%%%%%%%%%%%%%%%%%%%%%%%%%%%%
%%%%%%%%%%%%%%%%%%%%%%%%%%%%%%%%%%%%%%%%%%%%%%%%%%%%%%%%%%%%%%%%%%%%%%%%%%%%%%%%%%%%%%%%%%%%
\section{Introduction}
\label{sec:intro}

%%%%%%%%%%%%%%%%%%%%%%%%%%%%%%%%%%%%%%%%%%%%%%%%%%%%%%%%%%%%%%%%%%%%%%%%%%%%%%%%%%%%%%%%%%%%

%%%%%%%%%%%%%%%%%%%%%%%%%%%%%%%%%%%%%%%%%%%%%%%%%%%%%%%%%%%%%%%%%%%%%%%
%%%%%%%%%%%%%%%%%%%%%%%%%%%%%%%%%%%%%%%%%%%%%%%%%%%%%%%%%%%%%%%%%%%%%%%
\subsection{Summary of stabilizing reinforcement learning}

Reinforcement learning is an optimal control method that uses adaptation imitating living beings in their natural environments \cite{Lewis2013, Sutton2018, Bertsekas2019}.
%Its applications range from robotics \cite{Kumar2016,Borno2013,Tassa2012,Surmann2020,Akkaya2019} to games such as Go, chess, shogi (also known as Japanese chess) \cite{Silver2016,Silver2018}, even complex ones such as StarCraft II \cite{Vinyals2019}.
Its applications range from robotics \cite{Kumar2016,Surmann2020} to games such as Go, chess, shogi (also known as Japanese chess) \cite{Silver2018}, even complex ones such as StarCraft II \cite{Vinyals2019}.

Reinforcement learning can be model-based or model-free, online or offline etc.
Online here means being able to learn during exploitation, as opposed running training procedures beforehand.
Offline, on contrary, refers to learning by prior optimization of an agent that is only deployed once the learning is finalized.

Contemporary reinforcement learning methods are usually understood in the context of approximating (learning) the cost-to-go function $J$ (or Q-function, or advantage function $A$ or something else related to the cost-to-go) via (deep) neural networks.
The core problem with such an approach is that one cannot know how good the chosen neural network topology is capable of approximating the cost-to-go function.

\textit{Although it is known that the optimizer (the optimal policy) has certain desirable properties \eg it keeps the system stable, an optimizer resulting from an approximate cost-to-go function has in general no guarantees for the closed loop, first and foremost in terms of stability.}

Measures were and are taken to provide the said guarantees.
These roughly go in the following three directions:
\begin{enumerate}
\item introduction of a filter to discard unsuitable actions.
Such a filter may be human-based \cite{Saunders2017} or designed on the grounds of formal verification \cite{Alshiekh2018};
\item merging of reinforcement learning with (stabilizing) model-predictive control \cite{Beckenbach2019,Beckenbach2018,Beckenbach2020,Zanon2019,Zanon2020,Berkenkamp2019};
\item Lyapunov-based reinforcement learning \cite{Perkins2001,Chow2018,Osinenko2020a}.
%\item merging of reinforcement learning with (stabilizing) model-predictive control \cite{Beckenbach2019,Beckenbach2018,Beckenbach2020,Zanon2019,Zanon2020,Berkenkamp2019};
%\item Lyapunov-based reinforcement learning \cite{Perkins2001,Perkins2002,Chow2018,Choi2020,Hewing2020,Osinenko2020a}.

\end{enumerate}
A detailed survey on these methods can be found in \cite{osinenko2022reinforcement}.

There also exists an approach that intends to utilize the cost-to-go function itself as an effective Lyapunov function for the closed loop.
This approach seems only viable if the learned policy is close enough to the optimal one, although some claim an actor-critic controller to be stabilizing under usage of model-free tools such as robustifying terms from the adaptive control field \cite{Vamvoudakis2015}.
Unfortunately, the latter paper contains erroneous results -- see the detailed analysis and a counter-example in \cite{Osinenko2021i}.

We focus on the Lyapunov-based stabilizing reinforcement learning as per the third direction as listed above.
Some approaches of this category, such as \cite{Perkins2001,Perkins2002,Hewing2020}, can be regarded as offline.
Our interest here lies in online control methods though.

%% Check the citations below !!!!!!!
In previous works by the author \cite{Goehrt2019b,Goehrt2019,Goehrt2020, Osinenko2020a}, a framework was developed for stabilizing, online, Lyapunov-based reinforcement learning. The results in \cite{Osinenko2020a} were enabled by the techniques of sample-and-hold stabilization analyses \cite{Clarke1997,Osinenko2018,Schmidt2021}, which were recently extended to the case of stochastic systems \cite{Osinenko2021d, Osinenko2021e}.

%%%%%%%%%%%%%%%%%%%%%%%%%%%%%%%%%%%%%%%%%%%%%%%%%%%%%%%%%%%%%%%%%%%%%%%%%%%%%%%%%%%%%%%%%%%%
\subsection{Stabilizing reinforcement learning via Lyapunov techniques}
\label{sub:stab-RL-LF}

Lyapunov theory and its stochastic generalization due to R. Khasminskii \cite{Khasminskii2011} are the recognized standard when it comes to stability analyses.
Already in the beginning of the 2000s, Lyapunov theory was used to design stabilizing reinforcement learning by Perkins and Barto \cite{Perkins2001, Perkins2002}.
The latter work meant to  design or restrict the action choices of the agent so that on of the Lyapunov theorems applies.
A variety such approaches were developed since.
For instance, \cite{Chow2018} suggested a stability-preserving Bellman operator ensure the current cost-to-go function estimate satisfy a Lyapunov condition.
The approach required iterations over a region in the state space and also that the baseline (stabilizing) policy be sufficiently close to the optimal one.
Berkenkamp et al. \cite{Berkenkamp2017} used state space splitting to verify a Lyapunov condition in each cell.
The approach needed certain confidence intervals on the statistical model of the environment.
Lyapunov techniques were employed in \cite{Vamvoudakis2015} along with elements of adaptive control.
This approach tacitly assumed the input-coupling function be bounded away from zero to effectively employ a kind of high-gain compensation of the sign-indefinite terms in the final Lyapunov condition (see Section IV in the said work, in particular, formula 45).
Such an assumption is fairly restrictive and the high-gain compensation might neutralize the effects of learning, while harming minimal interference.

Also, techniques of adaptive control were used in stabilizing approximate optimal control for environments with partially unknown models as per \cite{Kamalapurkar2016}.
In \cite{Vamvoudakis2015} and \cite{Kamalapurkar2016}, the cost-to-go function was effectively used as the Lyapunov function candidate for the state.
The drift was assumed linear in unknown parameters and the input-coupling function was assumed uniformly bounded.
An overview of such techniques, that employ adaptive control in reinforcement learning, can be found in \cite{Kamalapurkar2018}.

%%%%%%%%%%%%%%%%%%%%%%%%%%%%%%%%%%%%%%%%%%%%%%%%%%%%%%%%%%%%%%%%%%%%%%%%%%%%%%%%%%%%%%%%%%%%
%%%%%%%%%%%%%%%%%%%%%%%%%%%%%%%%%%%%%%%%%%%%%%%%%%%%%%%%%%%%%%%%%%%%%%%%%%%%%%%%%%%%%%%%%%%%
\section{Summary of the state of the art and contribution}
\label{sec:contribution}

This paper adds to the state of the art of Lyapunov techniques in stabilizing reinforcement learning.
The proposed method introduces specially designed Lyapunov-like constraints on the agent learning so that the system be stabilized in a suitable sense.
The overall framework implying the imposition of such constraints is titled \textbf{C}ritic \textbf{A}s a \textbf{L}yapunov \textbf{F}unction (\textbf{CALF}).
Indeed, it is shown that the critic becomes effectively a (strictly speaking, time-varying) Lyapunov function for the closed loop with a suitable decay condition.
%Essentially, The latter constraints ensure a Lyapunov-like behavior of the critic that in turn makes the origin of the agent-environment closed loop stable.

We consider environments influenced by bounded stochastic disturbance in a sample-and-hold setting \ie we assume the underlying dynamics be continuous, whereas the actions are taken at discrete moments in time.
The inter-sample behavior of the environment is addressed explicitly in the stability analysis.
The reason to assume bounded noise is because, otherwise, no guarantees can be made on the inter-sample behavior of the trajectories and this is independent of the agent design.
Yet, bounded noise is physically meaningful and we give an overview of some respective stochastic models in the appendix.

Unlike some of the existing Lyapunov-based approaches, the currently presented one is purely online.
It should be noted that the derived stability guarantees do not require persistence of excitation. 
The proposed approach applies to environments with general nonlinear dynamics, not necessarily control-affine, nor does it need the input-coupling to be bounded or bounded away from zero.

The proofs of all the theorems, as well as the extra graphical material on the simulation studies, are to be found in the appendix.

\subsection{Notation}
\label{sub:notation}

The set of real numbers with zero is denoted $\R_+$.
A Lipschitz constant of a locally Liphscitz continuous function $\varphi$ is denoted as $\lip{\varphi}$, whereas the concrete compact set on which this constant is taken is defined in the each specific context.
Spaces of class kappa functions are denoted $\K, \K_\infty$.
A diagonal matrix formed from a vector $x$ is denoted $\diag{x}$.

%%%%%%%%%%%%%%%%%%%%%%%%%%%%%%%%%%%%%%%%%%%%%%%%%%%%%%%%%%%%%%%%%%%%%%%%%%%%%%%%%%%%%%%%%%%%
%%%%%%%%%%%%%%%%%%%%%%%%%%%%%%%%%%%%%%%%%%%%%%%%%%%%%%%%%%%%%%%%%%%%%%%%%%%%%%%%%%%%%%%%%%%%
\section{Background}
\label{sec:background}

Consider an agent-environment loop that assumes a state $X_t \in \mathbb{X} \subset \R^n$ at a time $t \in \mathcal{T} = \R_+$, whose dynamics are modeled by the following transition law:
\begin{equation}
\begin{aligned}
	\label{eqn:sys}
	&\diff X_t = f(X_t, U_t)\diff t + \sigma(X_t, U_t)\diff Z_t,\\
	&\mathbb{P}\left[X_0 = \text{const}\right] = 1
\end{aligned}
\end{equation}
where an unknown noise $Z_t \in \R^d$ is modelled by an arbitrary Lipschitz-continuous random process with a fixed Lipschitz constant $\lip{Z}$, and $\sigma : \mathbb{X} \times \mathbb{U} \ra \R^{n\times d}$ is a bounded mixing matrix \ie  $\nrm{\sigma(x, u)}_{\text{op}} \leq \sigma_{\max}, \sigma_{\max} > 0$.
Here and after, when claims are made about sections of random processes, the respective conditions are meant to hold almost surely.
We assume that $f$ and $\sigma$ are locally Lipschitz continuous with respect to $X_t$.
%This gives us sufficient grounds to prove the following
%
%\begin{thm}
%    The equation (\ref{eqn:sys}) has a solution for arbitrary $W_t$.
%\end{thm}
Particular models for such a bounded noise as in the above description are given in Section \ref{sub:noise-models} of the appendix.

Let a policy $\rho : \mathbb{X} \times \mathcal{T} \ra \mathbb{U}$ determine the agent's action $U_t \in \mathbb{U}$ at time $t$ in the following way:
\begin{equation}
	\label{eqn:sh}
	U_t = \rho(X_{t \ - \ t \ \text{mod} \ \delta}),
\end{equation}
where $\delta$ is the digital controller sampling time of the present (such a setting is also known as ``sample-and-hold'' or S\&H).
Under such a digital control, at best only practical stability in the sense of the following definition can be achieved.

\begin{dfn}
	\label{def:pract-stab}
	A policy $\rho[X_0](\cdot)$ is said to practically semi-globally stabilize \eqref{eqn:sys} from $X_0$ until $s^{\ast}$ if, given $s^{\ast}<r<\infty$, there exists a $\bar \delta > 0$ \sut under \eqref{eqn:sh} a trajectory $x(t)$ with a sampling time $\delta \le \bar \delta$, starting at $X_0$ is bounded, enters $\mathcal{B}_r$ after a time $t'$, which depends uniformly on $\nrm{X_0},r$, and remains there for all $t \geq t'$.
\end{dfn}
In the following, for brevity, the wording ``semi-globally'' is omitted. 

%There can be a vast multitude of policies that succeed at practically stabilizing the environment, but when it comes to choosing the most desirable one, a great deal of considerations can come into play \eg energy conservation, wear reduction, avoiding unfavourable zones and so on.
%
%That's why an algorithm that synthesizes such a policy should ideally also account for costs associated with the trajectories this policy produces. Thus we introduce a cost function

We consider infinite horizon optimal control in the sense of cost-to-go rather than accumulated reward, although the two are just different and equipollent formulations. 
Let the running cost be $\rho : \mathbb{X} \times \mathbb{U} \ra \R$.
The aim of reinforcement learning is to find a policy that optimizes a cost-to-go $J_\rho : \mathbb{X} \ra \R$, where
$$
J_\rho(X_0) := \E{\int \limits_0^{\infty} r(X_t, U_t) \diff t}, \ X_0 = const \text{ almost surely},
$$
and $U_t$ satisfies \eqref{eqn:sh}.
The minimum of the cost-to-go is the \textit{optimal cost-to-go function} $J$.
 
Note, that despite the fact that $t$ appears in \eqref{eqn:sh}, $J_{\rho}$ need not depend on time, since $\rho$ is interpreted as a stationary policy in discrete time.

\begin{rem}
It must be noted that the proposed approach is also applicable to the case, when one is dealing with maximization of total reward rather than minimization of total cost. In the former case however running costs and cost-to-go function $J$ are replaced with running rewards and value function $V$ respectively, while the minima of total cost are obviously replaced with maxima of total reward.
\end{rem}

The goal of the present paper is to determine an algorithm capable of constructing a practically stabilizing policy on the fly, given the agent's transition law, while also addressing the above stated aim via reinforcement learning methods.
In other words, the proposed algorithm is intended to generate a quasi-optimal practically stabilizing control policy $\rho[X_0]$ for an arbitrary initial state $X_0$ given the state transition function $f(\cdot, \cdot)$.
Notably, the generation of such policies should not require any pre-training \ie the algorithm is supposed to immediately provide a viable policy given only $f(\cdot, \cdot)$ and $X_0$.

In the next section, we give the central assumptions in succeeding at this problem statement.

%%%%%%%%%%%%%%%%%%%%%%%%%%%%%%%%%%%%%%%%%%%%%%%%%%%%%%%%%%%%%%%%%%%%%%%%%%%%%%%%%%%%%%%%%%%%
\subsection{Assumptions}
\label{sub:assumptions}

The starting point in many reinforcement learning analyses is the requirement that \eqref{eqn:sys} be stabilizable \cite{AlTamimi2008,Heydari2014,Liu2013,Jiang2015}.
%(in non-S\&H mode \ie with $U_t = \rho(X_t, t)$)
In this work, we make the following assumptions that are both implied by stabilizability \cite{Hafstein2005}.

\begin{asm}
	\label{asm:CLF}
	There exists a locally Lipschitz continuous, positive-definite function $L:\R^n \ra \R_{+}$, a diverging continuous positive definite function $\nu:\R^n \ra \R_{+}$ and $\alphalow, \alphaup \in \Kinf$ \sut for any compact $\bar{\mathbb{X}} \subset \mathbb{X}$, there exists a compact set $\bar{\mathbb{U}} \subset \mathbb{U}$ and it holds that, for any $x \in \bar{\mathbb{X}}$, 
	\begin{enumerate}
		\item[i)] $L$ has a decay rate satisfying
		\begin{align} \label{eqn:CLF-decay}
		\inf_{u \in \bar{\mathbb{U}}} \; \mathcal{L}_{f(x,u)} L(x) \leq - \nu_{0} (x),
		\end{align}
		\item[ii)] $\alphalow(\nrm{x}) \leq L(x) \leq \alphaup(\nrm{x})$.
		%		\item[iii)] for all $y \in \X$, 
		%		\begin{align} \label{eqn:CLF-Lipschitz}
		%			\|L(x) - L(y) \| \leq \lip{L} \|x - y\|.
		%		\end{align}
	\end{enumerate}
\end{asm} 
In the above, $\mathcal{L}_{f(x,u)}L(x)$ is a suitable generalized directional differential operator of $L(x)$ in the direction of $f(x,u)$ \eg directional subderivative \cite{Clarke2008}.
\begin{asm}
\label{asm:sample-wise-decay}

There exists a stabilizing sample-and-hold policy $\eta : \X \rightarrow \U$ associated with the latter Lyapunov function $L$ \ie
\begin{multline}
\forall \delta \in (0, \bar \delta] \ \ : \\
U_{t} := \eta(X_{t \text{ mod } \delta}) \implies \\  L(X_{(k + 1)\delta}) - L(X_{k\delta}) \leq -\delta \nu(X_{k\delta}),
\end{multline}
where $\nu(x) > 0$ whenever $\bar S_{\sigma_{\max}, \bar \delta} > \nrm{x} > \bar s_{\sigma_{\max}, \bar \delta}$.
\end{asm} 

A recent stochastic stability analysis \cite{Osinenko2021e} has shown that as long as the noiseless system
\begin{multline}
\diff X^{\textit{noiseless}}_{t} = f(X^{\textit{noiseless}}_{t}, U_{t})\diff t\\
\end{multline}
is stabilizable, assumptions 1-2 hold. 
Furthermore it is true that 
%\begin{multline}
%\bar s_{\sigma_{max}, \bar \delta} \rightarrow 0, \text{ as } {{\sigma_{\max} \rightarrow 0} \atop {\bar \delta \rightarrow 0}},\\
%\bar S_{\sigma_{max}, \bar \delta} \rightarrow +\infty, \text{ as } {{\sigma_{\max} \rightarrow 0} \atop {\bar \delta \rightarrow 0}}.\\
%\end{multline}

\begin{multline}
\bar s_{\sigma_{max}, \bar \delta} \rightarrow 0, \text{ as } \sigma_{\max} \rightarrow 0 \text{ and } \bar \delta \rightarrow 0,\\
\bar S_{\sigma_{max}, \bar \delta} \rightarrow +\infty, \text{ as } \sigma_{\max} \rightarrow 0 \text{ and } \bar \delta \rightarrow 0.\\
\end{multline}

% (until an $s^{\ast}_L$ that depends on $w, L, W_{\max}, \sigma_{\max}$) \sut
%\begin{equation}
%	\label{eqn:decay-nom-policy}
%	x=1
%\end{equation}

\begin{rem}
	\label{rem:LF}
	Assumption \ref{asm:CLF} is weaker than in some reinforcement learning schemes with guarantees, such as \cite{Jiang2015} where $L$ needs actually to satisfy a Hamilton-Bellman-Jacobi inequality, but is still somewhat strong (however, see \cite{Berkenkamp2017}).
	This is, as discussed above, the price to pay for ensuring practical stabilizability property of the policy resulting from learning.
%	However, already controllability, as required in convergence analyses of reinforcement learning (see \eg \cite{AlTamimi2008,Heydari2014,Liu2013,Jiang2015}), implies stabilizability \cite{Sastry2013,Clarke1997}.
	However, already controllability, as required in convergence analyses of reinforcement learning (see \eg \cite{Heydari2014,Jiang2015}), implies stabilizability \cite{Sastry2013,Clarke1997}.
	If the closed-loop is stable, there is \eg a constructive way to find a Lyapunov function (see \cite{Hafstein2004}).
%	Also, there are numerical techniques to build control Lyapunov functions \cite{Baier2014,Giesl2015a,Baier2012}.
	Also, there are numerical techniques to build control Lyapunov functions \cite{Baier2014}.	
	Furthermore, they can be constructed from physical considerations \eg kinetic and potential energy \cite{Berkenkamp2017}.
	Then, there are systematic techniques such as augmented Lyapunov functions in adaptive control; backstepping; stabilizing model-predictive control whose finite-horizon cost function can be shown to be a Lyapunov function etc.
\end{rem}

\begin{rem}
Also note that the existence of such $\alphalow(\cdot)$ and $\alphaup(\cdot)$  are equivalent to radial unboundedness and positive-definiteness respectively \cite{Khalil1996}. Furthermore, a constructive proof of the latter properties would have to involve a construction of such functions. Thus if a $L(\cdot)$ was constructively proven to be a Lyapunov function, then the respective $\alphalow(\cdot)$ and $\alphaup(\cdot)$ can be extracted from the latter proof.
\end{rem}

Let the critic $\hat{J} : \mathbb{X} \times \W \ra \R$ be an approximator \eg a deep neural network, with weights $w \in \W$, where $\W$ is compact and $0 \notin \W$.
It is up to the user of the proposed approach to choose their critic structure.
A flexible critic, capable of closely approximating the cost-to-go function will likely yield a better performance.
However, the stabilization itself will occur regardless of such nuances as long as the critic is chosen so as to satisfy the following assumption:
The actor will be considered also in a parametric form $\rho^\theta : \X \times \Theta \ra \U$, where $\Theta$ is a compact actor weight set.

\begin{asm}
	\label{asm:struct-equiv}
	For each $\zeta > 0$ there exists $w^{\#}_{\zeta} \in \W$ \sut $\zeta L(x) = \hat{J}^{w^{\#}_{\zeta}}(x)$, for all $x \in \mathbb{X}$.
\end{asm} 
Note that if $\zeta > 0$ then $\zeta L(x)$ is also a control-Lyapunov function that corresponds to the same stabilizing policy as $L(x)$.

\begin{rem}
	\label{rem:struct-equiv-simple}
	The simplest way to satisfy Assumption \ref{asm:struct-equiv} is to choose $\hat{J}^{w}(x) = w_0 \varphi(x, w_{[2:N_c + 1]}) + w_1 L(x)$, where $\varphi$ is a (deep) neural network with some $N_c$  weights $w_{[2:N_c + 1]}$, and $w_0, w_1 \in [0, \infty)$.
\end{rem}

\begin{rem}
	\label{rem:struct-equiv}
	Assumption \ref{asm:struct-equiv} means that the critic structure be rich enough so as to capture the structure of the Lyapunov function as per Assumption \ref{asm:CLF}.
	We use it in such an exact format only for brevity.
	A relaxed form with some discrepancy between $\hat{J}^{w^{\#}_{\zeta}}(x)$ and $\zeta L(x)$ would also do the job, one would only to account for it in the constraints \eqref{eqn:actor-critic-stab-1}--\eqref{eqn:actor-critic-stab-4}.
	A weaker, local condition of the kind $\forall \text{ compact } \mathbb X' \spc \exists w^{\#}_{\zeta} \spc \forall x \in \mathbb{X'} \spc \hat{J}^{w^{\#}_{\zeta}}(x) = \zeta L(x)$ would also fit, whereas one would just need to determine such an $\mathbb X'$ that the current and predicted state lie in.
	Again, we choose the simplest form for the ease of exposition.
\end{rem}

\begin{asm}
The critic $\hat J^w(\cdot)$ is locally Lipschitz continuous for every $w \in \W$. Thus the notation $\lip{\hat J}$ implies $\forall w \in \W \ \lvert\hat J^{w}(x) - \hat J^{w}(x') \rvert \leq \lip{\hat J}\nrm{x - x'}$. Though, since $\hat J(\cdot)$ is \textbf{locally} Lipschitz continuous, its Lipschitz constant $\lip{\hat J}$ is ambiguous until a bounded domain is specified.
\end{asm}
%%%%%%%%%%%%%%%%%%%%%%%%%%%%%%%%%%%%%%%%%%%%%%%%%%%%%%%%%%%%%%%%%%%%%%%%%%%%%%%%%%%%%%%%%%%%
%%%%%%%%%%%%%%%%%%%%%%%%%%%%%%%%%%%%%%%%%%%%%%%%%%%%%%%%%%%%%%%%%%%%%%%%%%%%%%%%%%%%%%%%%%%%
\section{Preliminaries}
\label{sec:preliminaries}

Let $\Lambda_{U_t}(X_t)$ be the Euler estimate of $X_{t + \delta}$ for the noiseless system:
$$
\Lambda_{U_t}(X_t) = X_t + \delta f(X_t, U_t).
$$

\begin{rem}
	\label{rem:prediction scheme}
	The Euler prediction scheme is chosen for the ease of exposition of the forthcoming analyses, but other choices, including Runge-Kutta methods would work.
\end{rem}

Online reinforcement learning implies that at each time step $k$ an update is being carried out for both the control input $u$ and the critic weights $w$:
\begin{equation}
u_k, w_k \leftarrow F(x_k, w_{k-1}, \dots)
\end{equation}
The present paper discusses how online stability can be ensured by enforcing certain stabilzing constraints when performing the latter update. Let the stabilizing constraint function $G^{\bar \nu}_\eps(\cdot, \cdot, \cdot, \cdot)$ be defined as follows:
\begin{equation}
G^{\bar \nu}_\eps(w, \wo, \xn, x) := \begin{bmatrix}
\hat{J}^w(x) - \hat{J}^\wo(x) - \delta \eps\\
L(\xn) - \hat{J}^w(\xn) \\
\hat{J}^w(\xn) - \hat{J}^w(x) + \delta {\bar \nu} \\
\hatalphalow(\nrm{x}) - \hat{J}^w(x)\\
\hat{J}^w(x) - \hatalphaup(\nrm{x})\\
\end{bmatrix}
\end{equation}
where $\hat{\alpha}_{1, 2} \in \mathcal{K}_\infty$ are chosen $\sut \forall x \spc \hatalphalow(\nrm{x}) \le \alphalow(\nrm{x}), \alphaup(\nrm{x}) \le \hatalphaup(\nrm{x})$.

Note that ensuring $G^{\bar \nu}_\eps(w, \wo, \xn, x) \leq 0$ is equivalent to imposing the following four scalar constraints: 
\begin{subequations}
	\label{eqn:actor-critic-stab}
	\begin{align}
	&\hat{J}^w(x) \leq \hat{J}^\wo(x) + \delta \eps,  \label{eqn:actor-critic-stab-1} \tag{C1}\\
	& L(\xn) \leq \hat{J}^w(\xn), \label{eqn:actor-critic-stab-2} \tag{C2}\\
	& \hat{J}^w(\xn) - \hat{J}^w(x) \leq - \delta {\bar \nu}, \label{eqn:actor-critic-stab-3} \tag{C3}\\
	& \hatalphalow(\nrm{x}) \leq \hat{J}^w(x) \leq \hatalphaup(\nrm{x}), \label{eqn:actor-critic-stab-4} \tag{C4}
	\end{align}
\end{subequations}

\section{Algorithm}
\label{sec:algorithm}

It will be shown that enforcing either $G^{\bar \nu}_\eps(w_k, w_{k-1}, x_k, x_{k - 1}) \leq 0$ or $G^{\bar \nu}_\eps(w_k, w_{k-1}, \Lambda_u(x_k), x_k) \leq 0$ is sufficient to guarantee a certain kind of stability. At the same time the latter constraints do not prevent the critic from converging to the optimal cost-to-go function.\\
Below are examples of particular algorithms that involve these constraints.

%\begin{itemize}

%\item

A SARSA-like approach can be used if the model of the noise is not priorly known is shown in Algorithm \ref{alg:CALF_SARSA}.

\begin{algorithm}[h]
	\SetAlgoLined
	\KwData{$\eps > 0$, ${\bar \nu} > 0$, $\zeta > 0$}
	\KwResult{$U_0, U_{\delta}, U_{2\delta}, \dots$}
	$x_0 \la X_0$;\\
	$w_0 \la w^\#_{\zeta}$;\\
	\For{$k = 0 \ \dots \ \infty$}
	{
	${\left({\theta_{k}} \atop {w_{-1}}\right)} \gets \argmin\limits_{\text{\sut} G^{\bar \nu}_\eps(w, w_k, \Lambda_{(\rho^{\theta}(x_k))}(x_{k}), x_{k}) \le 0 } \Jc(\theta, w) $ \\
	where
	{\setlength{\mathindent}{0cm}
	\begin{align*}
		\Ja(\theta, w) = & r(x_k, \rho^{\theta}(x_{k})) + \hat{J}^{w_k}(\Lambda_{\rho^{\theta}(x_{k})}(x_k))	
	\end{align*} 
	} \\		
	$U_{k\delta} \la \rho_{\theta_{k}}(x_{k})$ \\
	\textbf{perform} $U_{k\delta}$ \\
	$x_{k + 1} \gets$ \textbf{observe} $X_{(k + 1)\delta}$ \\				
	$w_{k + 1} \gets  \argmin\limits_{ \text{\sut} G^{\bar \nu}_\eps(w, w_k, x_{k + 1}, x_{k}) \le 0 } \Jc(w)$ \\
	where
	{\setlength{\mathindent}{0cm}
	\begin{align*}
		\Jc(w) = & \left(\hat{J}^{w}(x_k) - r(x_k, U_{k\delta}) - \hat{J}^{w_k}(x_{k + 1})\right)^2
	\end{align*} 
	}
	}		
	\caption{CALF as SARSA}
	\label{alg:CALF_SARSA}
\end{algorithm}

%\item 

A sort of an actor-critic can be employed provided that one has the means to estimate the expected value (for instance by computing the mean of an artificial sample).
It is shown in Algorithm \ref{alg:CALF_AC}.

\begin{algorithm}[h]
	\SetAlgoLined
	\KwData{$\eps > 0$, ${\bar \nu} > 0$, $\zeta > 0$}
	\KwResult{$U_0, U_{\delta}, U_{2\delta}, \dots$}
	$w_0 \la w^\#_{\zeta}$;\\
	\For{$k = 0 \ \dots \ \infty$}{	
		$x_{k} \gets$ \textbf{observe} $X_{k\delta}$\\
		$ \left({{\theta_{k}} \atop {w_{-1}}}\right) \gets \argmin\limits_{\text{\sut} G^{\bar \nu}_\eps(w, w_k, \Lambda_{(\rho^{\theta}(x_k))}(x_{k}), x_{k}) \le 0 } \Ja(\theta, w)$ \\
		where
		{\setlength{\mathindent}{0cm}
		\begin{align*}
			\Ja(\theta, w) & = r(x_k, \rho^{\theta}(x_{k})) + \\
			& \E{\hat{J}^{w_{k}}(X_{(k + 1)\delta}) | X_{k\delta} = x_k, U_{k\delta}=\rho^{\theta}(x_{k})} \\
		\end{align*} 
		}			
		$w_{k + 1} \gets \argmin\limits_{\text{\sut} G^{\bar \nu}_\eps(w, w_k, \Lambda_{(\rho^{\theta_k}(x_k))}(x_{k}), x_{k}) \le 0 } \Jc(w) $ \\
		where
		{\setlength{\mathindent}{0cm}
		\begin{align*}
			\Jc(w) & =  \Big( \hat{J}^{w}(x_k) - r(x_k, U_{k\delta}) - \\
			& \E{\hat{J}^{w_k} ( X_{ \delta(k + 1) } )| X_{k\delta} = x_k, U_{k\delta}=\rho_{\theta_{k}}(x_{k})} \Big)^2
		\end{align*} 
		}	
		$U_{k\delta} \gets \rho_{\theta_{k}}(x_{k})$ \\	
		\textbf{perform} $U_{k\delta}$ \\				
	}
	\caption{CALF as stabilizing actor-critic}
	\label{alg:CALF_AC}
\end{algorithm}

%\end{itemize}

Although the above examples each present a case of single step temporal difference learning, the proposed approach would not in any way inhibit a different choice of temporal difference updates, such as for instance TD($N$) or TD($\lambda$). In fact the stabilizing properties of an algorithm of this kind are not affected by how the updates are performed as long as the said updates do not violate the stabilizing constraints. Though, of course, these choices do affect the resulting total cost.

\begin{rem}
	\label{rem:regularization}
	One can introduce regularization terms into the actor and critic loss functions with some tuning parameters $\beta_{1,2,3} > 0$ as follows:
	\begin{align*}
		& \Ja(\theta, w) + \beta_{1} \nrm{w - w_{k}}^2 + \beta_{2}\nrm{\theta - \theta_{k - 1}}^2 \\	
		& \Jc(w) + \beta_{3} \nrm{w - w_{k}}^2.	
	\end{align*}
	These regularization terms may be employed to regulate weight change thus imitating a learning rate of a gradient-like update rules.
\end{rem}

\begin{rem}
	\label{rem:critc_loss_exp_replay}
	One can formulate the loss \eg in terms of an experience replay of size $M$:
	\begin{equation}
		\label{eqn:critc_loss_exp_replay}
		w_{k + 1} \gets  \argmin\limits_{w \in \{ w' \in \W \ | \ G^{\bar \nu}_\eps(w', w_k, x_{k + 1}, x_{k})\leq 0\} } \sum_{j=0}^{M - 1}e_{j}(w) ;
	\end{equation}
	where $e_j(w) = \big(\hat{J}^{w}(x_{k - j}) -  r(x_{k - j}, U_{k\delta}) - \hat{J}^{w_k}(x_{k + 1 - j})\big)^2$ and the buffer count is backwards relative to $k$.
\end{rem}

\begin{rem}
	\label{rem:discounting}
	The above described loss functions are formulated inspired by the Bellman optimality principle for undiscounted problems.
	Introducing a $\gamma \in (0, 1)$ in front of the predicted $\hat J$ would make it to address a discounted problem instead.
\end{rem}

\begin{rem}
	\label{rem:need_known_LF}
	Notice that Algorithms \ref{alg:CALF_SARSA} and \ref{alg:CALF_AC} explicitly involve $L$, if it is available.
	Another variant of CALF presented in Algorithm \ref{alg:no-lyapunov} lifts this requirement, although both require knowledge of a stabilizing policy. 
\end{rem}

\begin{rem}
	\label{rem:safety-checker}
	An immediate hint on the practical implementation of the approach should be made.
	As Theorem \ref{thm:feas} shows, the problem \eqref{eqn:actor-critic-stab} is feasible at every step, up until the system's state reaches a certain neighbourhood of the origin.
	Thus, one does not necessarily have to run constrained optimization algorithms.
	Instead, one might perform unconstrained optimization and check constraint \eqref{eqn:actor-critic-stab-1}--\eqref{eqn:actor-critic-stab-4} satisfaction afterwards and reset to the initial admissible solution in case of violation.
	In this case, the scheme is somewhat reminiscent of the supervisor-based stabilizing reinforcement learning.
	In fact, one can just add penalty terms involving the constraints to the loss function.
	These considerations can help speed up the calculations.
\end{rem}

\begin{rem}
	\label{rem:min-inference}
	A remark should be made about the minimal inference aspect.
	The constraints \eqref{eqn:actor-critic-stab-1}--\eqref{eqn:actor-critic-stab-4} in general restrict the learning of the agent, which is evident.
	However, it must be pointed out that if $\eps_{4}$, ${\bar \nu}$, $\alpha_{\text{low}}$ and $\alpha_{\text{up}}$ are appropriately configured, the optimal critic weights will not be excluded from the set of attainable critic weights. 
	Such a configuration, however  may fail to exist if the magnitude of noise is too large \ie $\sigma_{\max} \geq  \frac{\min\limits_{s^{\ast} \leq \nrm{x} \leq S^{\ast}}r(x, \rho^{\ast}(x))}{\lip{\hat{J}}\lip{Z}}$, where $\rho^{\ast}$ is the optimal policy. 
	The latter can however be remedied by introducing a scaling factor for the cost itself \ie $r(\cdot, \cdot) \gets a r(\cdot, \cdot)$, where $a > 1$. 
	This would obviously yield an equivalent problem that can be tuned appropriately.
\end{rem}

We can now state the major results about the CALF framework.
The present paper proves that imposing the aforementioned constraints will ensure that the system's state $X_{t}$ will eventually reach a certain neighbourhoods of the origin $\ball_{s^{\ast}}$. The latter radius $s^{\ast}$ is given by:
\begin{align}
\begin{split}
\label{eqn:until-original}
s^{\ast} = \inf \{s \in \mathbb{R} \ | \  &\inf_{s \leq \nrm{x} \leq S^\ast} \nu(x) > 6\lip{\hat{J}}\sigma_{\max}\lip{Z}
 \},
\end{split}
\end{align}
where $S^{\ast}$ is chosen in such a way that $\bar S > S^\ast > \hat{\alpha}_{\text{low}}^{-1}(\alphaup(\nrm{X_0}))$ and the respective Lipschitz constant $\lip{\hat{J}}$ is taken over $\B_{S^\ast}$. For sufficiently small $\bar \delta, \sigma_{\max}$ all of the above values are well defined.

\textit{Theorem 1} states that satisfying the constraints indeed results in such stabilization, while \textit{Therorem 2} states the constraints are in fact guaranteed to be satisfiable up until the system's state reaches $\ball_{s^{\ast}}$.
\begin{thm}
\label{thm:stab}
For a sufficiently small $\delta > 0$ the following implication holds:
\begin{multline}
\label{eqn:implication-in-thm1}
	\forall k \in \mathbb{N}_{0}  \ G^{{\bar \nu}}_\eps(w_k, w_{k-1}, \Lambda_{U_{k\delta}}(X_{k\delta}), X_{k\delta}) \leq 0 \\ \lor \ X_{k\delta} \in \ball_{s^{\ast}} \\
	\implies \\
\inf\{t \in [0, +\infty) \ | \  X_{t} \in \ball_{s^{\ast}}\} < +\infty, \\	
	\end{multline}
	where $0 < \eps < {\bar \nu} - (1 + \frac{1}{3})\lip{\hat{J}}\sigma_{\max}\lip{Z}$.
\end{thm}

\begin{thm} 
\label{thm:feas}
For a sufficiently small $\delta > 0$ there exists such $\bar \nu$ that the following implication holds:
\begin{multline}
	\forall k < K  \ G^{\bar{\nu}}_\eps(w_k, w_{k-1}, \Lambda_{U_{k\delta}}(X_{k\delta}), X_{k\delta}) \leq 0 \ \lor \ X_{k\delta} \in \ball_{s^{\ast}} \\
	\implies \\
		\exists w \in \W, u \in \U   \ G^{\bar{\nu}}_\eps(w, w_{K - 1}, \Lambda_{u}(X_{K\delta}), X_{K\delta}) \leq 0 \ \\ \lor \ X_{K\delta} \in \ball_{s^{\ast}} \\
	\end{multline}
where $0 < \eps < \bar \nu - (1 + \frac{1}{3})\lip{\hat{J}}\sigma_{\max}\lip{Z}$.
\end{thm}

Here, $s^{\ast}$ is the distance from the origin within which we can no longer guarantee feasibility of \eqref{eqn:actor-critic-stab}.
Naturally, in the presence of noise with some positive amplitude there will be such a neighbourhood of the origin, where it is no longer possible to reliably reduce the distance to the origin.
It should be noted here that disturbance attenuation methods might be used to augment the resulting policy to reduce the described noise effect, but it is beyond the scope of this work. 
The next theorem concerns feasibility of the CALF optimization problem.

The next section presents the experimental evaluation.

%%%%%%%%%%%%%%%%%%%%%%%%%%%%%%%%%%%%%%%%%%%%%%%%%%%%%%%%%%%%%%%%%%%%%%%%%%%%%%%%%%%%%%%%%%%%
%%%%%%%%%%%%%%%%%%%%%%%%%%%%%%%%%%%%%%%%%%%%%%%%%%%%%%%%%%%%%%%%%%%%%%%%%%%%%%%%%%%%%%%%%%%%
\section{Experiments}
\label{sec:experiment}

%%%%%%%%%%%%%%%%%%%%%%%%%%%%%%%%%%%%%%%%%%%%%%%%%%%%%%%%%%%%%%%%%%%%%%%%%%%%%%%%%%%%%%%%%%%%
\subsection{Preliminaries of experimental validation}
\label{sub:experiment-prelim}

Experiments were performed with a mobile robot under CALF in Algorithm \ref{alg:CALF_SARSA} variation with the critic loss defined via temporal difference over an experience replay as per \eqref{eqn:critc_loss_exp_replay}.
The noiseless dynamics part of the robot were described as:
\begin{equation}
	\label{eqn:sys-3wrobot}
	\diff X_t = f_{\text{Cart}}(X_t, U_t)\diff t,
\end{equation}
where $T_{\text{Cart}}(x_1, x_2, x_3, u_1, u_2) = \left( u_1 \cos x_3, u_1 \sin x_3, u_2 \right)^\top$, $x_1, x_2, x_3$ are the Cartesian coordinates and angle, and $u_1, u_2$ are the linear and angular speed, respectively.
Both actor and critic were implemented as neural networks.
The running cost was considered in the following form:
\begin{equation}
\begin{aligned}
	& \rho = \chi^\top H \chi, \\
\end{aligned}
\end{equation}
where $\chi = [x-x^*, u]$, $H$ is a diagonal, positive-definite, running cost matrix and $x^*$ is the target position (see details in the next section). 
The running cost matrices were taken in several variants, to differently prioritize the state and action components of the running cost, namely:
\begin{equation}
\label{eqn:running-cost-matrices}
	\begin{aligned}
		& H_1 = \diag{10, 10, 10, 0, 0}, H_2 = \diag{10, 10, 10, 1, 1}, \\
		& H_3 = \diag{10, 10, 100, 0, 0}, H_4 = \diag{10, 100, 10, 0, 0}, \\
		& H_5 = \diag{100, 10, 10, 0, 0}. \\
	\end{aligned}
\end{equation}
%The penalty term was implemented in an indicator manner as follows: 
%\begin{equation}
%\begin{aligned}
% 	\begin{cases}
%		& 0, \text{ if constraints C1--C4 are satisfied}, \\
%		& \Lambda \cdot \sum_{i=1}^{4}(K_i(x, w, \wo, \vartheta) - \eps_i), \text{ otherwise},
%	\end{cases}
%\end{aligned}
%\end{equation}
%where $K_{i}$ is the actual value of each $(Ci)$th constraint of \eqref{eqn:actor-critic-stab-1}--\eqref{eqn:actor-critic-stab-4}, and $\Lambda$ is a penalty parameter, chosen relatively big.

The model of the mobile robot is also known as the nonholonomic integrator (NI) $(\dot x_{\text{NI}1}, \dot x_{\text{NI}2}, \dot x_{\text{NI}3})^\top = (u_{\text{NI}1}, u_{\text{NI}2}, x_{\text{NI}1} u_{\text{NI}2} - x_{\text{NI}2} u_{\text{NI}1} )^\top $, for which various stabilization methods exist.
We took $L(x_{\text{NI}}) = \min \limits_{\zeta \in [-\rho, \rho]} \left\{ x_{\text{NI}1}^4 + x_{\text{NI}2}^4 + \tfrac{\abs{x_{\text{NI}3}}^3}{(x_{\text{NI}1} \cos (\zeta) + x_{\text{NI}2} \sin(\zeta) + \sqrt{\abs{x_{\text{NI}3}}})^2} \right\}$ and the corresponding nominal parking policy as per \cite{Kimura2015}, while using a transformation between the NI and Cartesian coordinates of \eqref{eqn:sys-3wrobot} was used.

It should be noted that the goal of the experimental validation was to physically demonstrate the capabilities of the CALF algorithm to improve accumulated cost compared to the nominal policy under various $H$ matrices.
The goal here was not to come up with the best parking controller which can be done by numerous ways, including model-predictive control, sliding-mode control etc.

The software core of this work is the specially designed and programmed Python package for hybrid simulation of agents and environments (generally speaking, not necessarily reinforcement learning agents) called \texttt{rcognita} \cite{rcognita2020}.
Its main idea is to have an explicit implementation of sample-and-hold controls with user-defined sampling time specification.
%A detailed description is given in Section \ref{sub:rcognita} of the appendix.
For the purposes of this study, \texttt{rcognita} was coupled with the Robot Operating System (ROS) to communicate with a \textit{Robotis TurtleBot 3}, the hardware platform used (details follow in the next Section).
Figure \ref{fig:rcognita-ROS} illustrates a flowchart of interconnection of \texttt{rcognita} and ROS.

\begin{figure}[h]
	\centering
	\includegraphics[width=0.99\columnwidth]{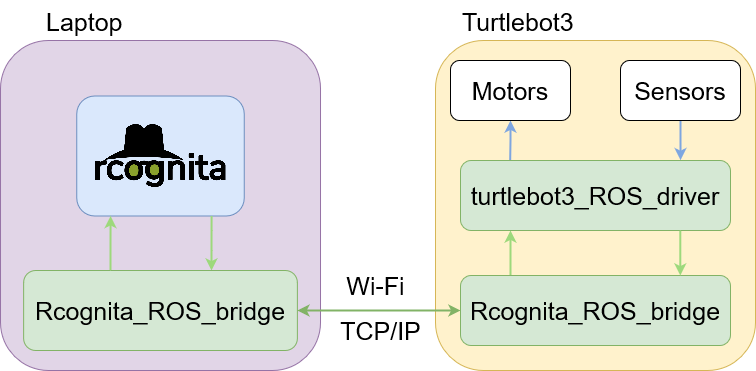}
	\caption{Flowchart of the interconnection of \texttt{rcognita} and Robotis TurtleBot 3, the robot used in the current study.}
	\label{fig:rcognita-ROS}
\end{figure}

%%%%%%%%%%%%%%%%%%%%%%%%%%%%%%%%%%%%%%%%%%%%%%%%%%%%%%%%%%%%%%%%%%%%%%%%%%%%%%%%%%%%%%%%%%%%
\subsection{Hardware and experiment setup}
\label{sub:experim-setup}

The hardware platform of the current study was a mobile wheeled robot Robotis TurtleBot 3 equipped with a lidar, an inertia measurement unit (IMU) for implementation of the linear and angular speed control, as well as for dead reckoning, which is also fused with the lidar data for position determination.

Experiments were run on a test polygon with concrete floor with markings of coordinate axes and 20 cm step nodes.
The robot started in the center with a fixed orientation and drove to on of the target positions as shown in Fig. \ref{fig:setup} (a).

\begin{figure*}
	\centering
	\begin{subfigure}[b]{0.45\textwidth}
		\centering
		\includegraphics[width=0.75\columnwidth]{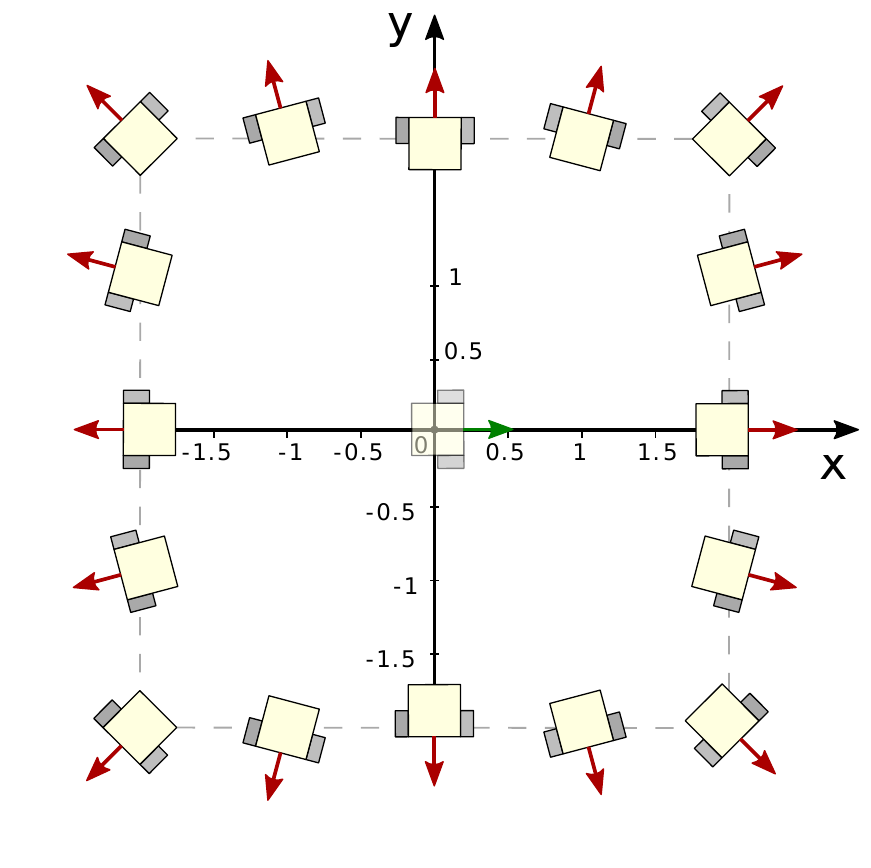}
		\caption{Target positions and orientations (marked \red{red}) for the robot. The initial orientation is marked \green{green}. The units are meters.}
	\end{subfigure}
	\hspace{10pt}
	\begin{subfigure}[b]{0.45\textwidth}
		\centering
		\includegraphics[width=0.75\textwidth]{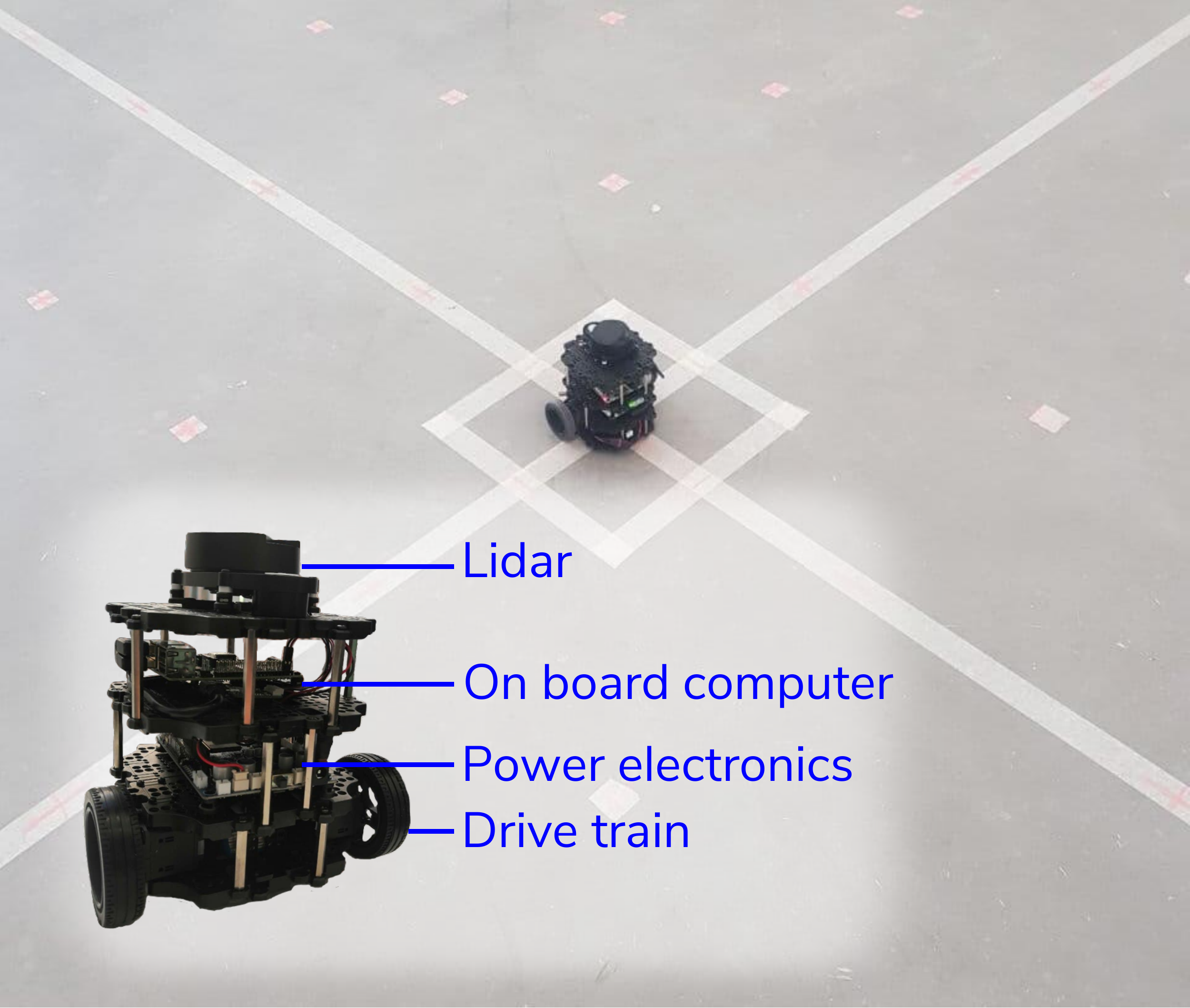}
		\caption{Robotis TurtleBot3 and its key components on the experiment polygon.}
	\end{subfigure}
	\caption{Overview of the experimental setup.}
	\label{fig:setup}
\end{figure*}

For each running cost matrix $H$, 32 runs were performed \ie for each target position and, respectively, the nominal and CALF agents.
The agent performance was evaluated by the accumulated running cost of 120 s, the total run time.
The agent sampling time was set to 0.05 sec which was fairly sufficient for the real-time robot control.
The results are presented in the next section.

%%%%%%%%%%%%%%%%%%%%%%%%%%%%%%%%%%%%%%%%%%%%%%%%%%%%%%%%%%%%%%%%%%%%%%%%%%%%%%%%%%%%%%%%%%%%
\subsection{Results}
\label{sub:results}

\begin{figure}
	\centering
	\includegraphics[width=0.99\columnwidth]{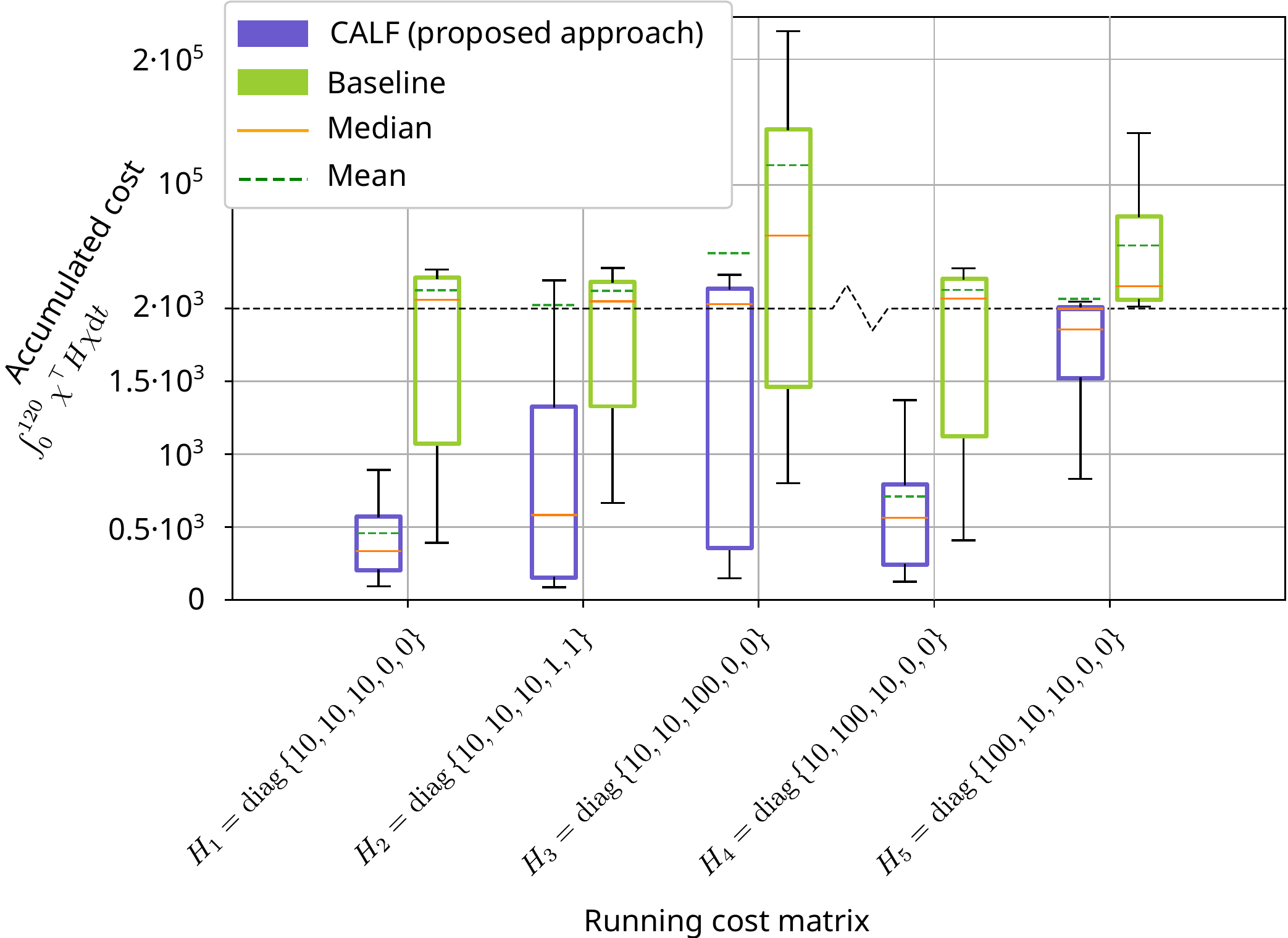}
	\caption{The accumulated running cost of the agents \ie $\int_0^{120} \chi^\top H \chi \diff t$ (see Section \ref{sub:experiment-prelim}).
	The box bounds are, respectively, the first and third quartiles Q1, Q3.
	The whiskers are the same quartiles plus/minus one and a half interquartile range Q3-Q1.
	The y-axis scale is split for better reading.}
	\label{fig:CALF-vs-nominal-accum-cost}
\end{figure}

The statistics of the accumulated running cost as per $\int_0^{120} \chi^\top H \chi \diff t$ for different running cost matrices under the CALF and nominal agents are shown in Fig. \ref{fig:CALF-vs-nominal-accum-cost}.
Each statistic was computed over multiple runs for each agent type and running cost matrix (see Fig. \ref{fig:setup}).
We observed that the CALF agent yielded considerably lower accumulated running cost showing statistically significant improvement under all selected running cost matrices.

\section{Discussion and extensions}

It would be fair to point out that in mechanical control engineering it is not uncommon to be aware of a stabilizing policy, while not knowing the corresponding Lyapunov function. However it is not difficult to adapt the proposed approach to a scenario of this kind. For instance, consider \textit{Algorithm \ref{alg:no-lyapunov}}.

\begin{algorithm}
\KwIn{${\bar \nu} > 0, \zeta > 0, \eta(\cdot) \text{ -- stabilizing policy}$}

$x_{0} \gets X_{0}$;\\
$w_{0} \gets $ \textbf{arbitrary};\\
$u_{0} \gets  \eta(x_{0})$;\\
$w_{\text{prev}} \gets w_{0}$;\\
$x_{\text{prev}} \gets x_{0}$;\\

\For{$k := 1...\infty$}{
%	Compute critic change before critic update: \\
%	$\Delta \hat J^{w_{k-1}}_k := \hat J^{w_{k-1}}(x_k) - \hat J^{w_{k-1}}(x_{k-1})$ \\
    \textbf{perform} $U_{(k - 1)\delta}$;\\
    $x_{k} \gets$ \textbf{observe} $X_{k\delta}$;\\
	Update the critic: \\
	{\setlength{\mathindent}{0cm}	
	\[	
	\begin{array}{lll}
		& w^* \gets &  \argmin \limits_{w \in \mathbb W}  \big(\hat{J}^{w}(x_{k-1}) - r_{\text{prev}} - \hat{J}^{\wo}(x_{k})\big)^2 \\
		& & \sut  \quad \substack{ \hat J^w(x_{k}) - \hat J^{w_{\text{prev}}}(x_{\text{prev}}) \leq -{\bar \nu} \delta,\\
		 \hat \alpha_{\text{low}}(\nrm{x_{k}}) \leq \hat J^{w}(x_{k}) \leq \hat \alpha_{\text{up}}(\nrm{x_{k}}) ;}
	\end{array}
	\]
	} \\	
	\eIf{ solution $w^*$ found}
	{	
		$w_{k} \gets w^{*}$\\
		Update the actor: \\
	{\setlength{\mathindent}{0cm}			
	\[
	\begin{array}{lll}
		& U_{k \delta} \la \argmin\limits_{u \in \U} \left\{r(x_k, u) + \hat{J}^{w_k}(\Lambda_u(x_k))\right\};
	\end{array}		
	\] 
	} \\
	$x_{\text{prev}} \gets x_{k}$\\
	$w_{\text{prev}} \gets w_{k}$\\
	$r_{\text{prev}} \gets r(x_{k}, u_{k})$
	}
	{
	Invoke a stabilizing action: \\
	$U_{k \delta} \la \eta(x_{k})$ \\
	$r_{\text{prev}} \gets r(x_{k}, u_{k}) + r_{\text{prev}}$;		
	}	
}
\caption{Stabilizing-policy-only CALF
%The actor's loss is denoted $\J^a$.
}
\label{alg:no-lyapunov}
\end{algorithm}

The above algorithm relies on similar considerations to those that validate the original results. Though instead of imposing (C2) it instead falls back to the stabilizing policy to maintain decay whenever the critic-based update fails to do so. This way we're merely performing  a kind of on-policy learning, except unlike SARSA the policy that is being evaluated does not perform exploratory actions, but instead sometimes performs stabilizing actions when necessary.

Much like in the case of the original approach, the algorithm ensures a certain kind of stability.

\begin{thm}
\label{thm:CALF-2}
For a sufficiently small $\delta > 0$ there exists such $\bar \nu$, that \textit{Algorithm \ref{alg:no-lyapunov}} will ensure 
\begin{multline}
\inf\{t \in [0, +\infty) \ | \  X_{t} \in \ball_{s^{\ast}}\} < +\infty. \\	
\end{multline}
\end{thm}

A brief testing procedure was too performed for \textit{Algorithm \ref{alg:no-lyapunov}} that involved the same set-up as described in section \ref{sec:experiment}. The performance of the algorithm was compared against that of both the nominal policy and the non-stabilizing version of \textit{Algorithm \ref{alg:no-lyapunov}} \ie the result of stripping constraints and nominal stabilizing actions from the latter. \textit{Table \ref{tbl:extra}} indicates a significant improvement as compared to both baselines.

\begin{table}
\caption{Average total cost}
\label{tbl:extra}
\begin{tabular}{|p{0.75\columnwidth}|p{0.15\columnwidth}|}
\hline
\textbf{Algorithm} & \textbf{Total cost} \\
\hline
{Stabilizing-policy-only CALF (\textit{Algorithm \ref{alg:no-lyapunov}})} & 88.0 \\
\hline
Actor-critic without stabilizing constraints & 152.0 \\
\hline
Nominal stabilizing policy & 252.0 \\
\hline
\end{tabular}
\label{tab_results}
\end{table}

\bibliographystyle{IEEEtran}
\bibliography{
bib/AIDA
}

\begin{thebibliography}{10}
\providecommand{\url}[1]{#1}
\csname url@samestyle\endcsname
\providecommand{\newblock}{\relax}
\providecommand{\bibinfo}[2]{#2}
\providecommand{\BIBentrySTDinterwordspacing}{\spaceskip=0pt\relax}
\providecommand{\BIBentryALTinterwordstretchfactor}{4}
\providecommand{\BIBentryALTinterwordspacing}{\spaceskip=\fontdimen2\font plus
\BIBentryALTinterwordstretchfactor\fontdimen3\font minus
  \fontdimen4\font\relax}
\providecommand{\BIBforeignlanguage}[2]{{%
\expandafter\ifx\csname l@#1\endcsname\relax
\typeout{** WARNING: IEEEtran.bst: No hyphenation pattern has been}%
\typeout{** loaded for the language `#1'. Using the pattern for}%
\typeout{** the default language instead.}%
\else
\language=\csname l@#1\endcsname
\fi
#2}}
\providecommand{\BIBdecl}{\relax}
\BIBdecl

\bibitem{Lewis2013}
F.~L. Lewis and D.~Liu, \emph{Reinforcement learning and approximate dynamic
  programming for feedback control}.\hskip 1em plus 0.5em minus 0.4em\relax
  John Wiley \& Sons, 2013, vol.~17.

\bibitem{Sutton2018}
R.~S. Sutton and A.~G. Barto, \emph{Reinforcement Learning: An
  Introduction}.\hskip 1em plus 0.5em minus 0.4em\relax Cambridge, MA, USA: A
  Bradford Book, 2018.

\bibitem{Bertsekas2019}
D.~P. Bertsekas, \emph{Reinforcement learning and optimal control}.\hskip 1em
  plus 0.5em minus 0.4em\relax Athena Scientific Belmont, MA, 2019.

\bibitem{Kumar2016}
V.~{Kumar}, E.~{Todorov}, and S.~{Levine}, ``Optimal control with learned local
  models: Application to dexterous manipulation,'' in \emph{2016 IEEE
  International Conference on Robotics and Automation (ICRA)}, 2016, pp.
  378--383.

\bibitem{Surmann2020}
H.~Surmann, C.~Jestel, R.~Marchel, F.~Musberg, H.~Elhadj, and M.~Ardani, ``Deep
  reinforcement learning for real autonomous mobile robot navigation in indoor
  environments,'' 2020.

\bibitem{Silver2018}
D.~Silver, T.~Hubert, J.~Schrittwieser, I.~Antonoglou, M.~Lai, A.~Guez,
  M.~Lanctot, L.~Sifre, D.~Kumaran, T.~Graepel, T.~Lillicrap, K.~Simonyan, and
  D.~Hassabis, ``A general reinforcement learning algorithm that masters chess,
  shogi, and go through self-play,'' \emph{Science}, vol. 362, no. 6419, pp.
  1140--1144, 2018.

\bibitem{Vinyals2019}
O.~Vinyals \emph{et~al.}, ``Grandmaster level in {StarCraft} {II} using
  multi-agent reinforcement learning,'' \emph{Nature}, vol. 575, no. 7782, pp.
  350--354, 2019.

\bibitem{Saunders2017}
W.~Saunders, G.~Sastry, A.~Stuhlmueller, and O.~Evans, ``Trial without error:
  Towards safe reinforcement learning via human intervention,'' \emph{arXiv
  preprint arXiv:1707.05173}, 2017.

\bibitem{Alshiekh2018}
M.~Alshiekh, R.~Bloem, R.~Ehlers, B.~K{\"o}nighofer, S.~Niekum, and U.~Topcu,
  ``\BIBforeignlanguage{English}{Safe reinforcement learning via shielding},''
  in \emph{\BIBforeignlanguage{English}{Proceedings AAAI}}, 2018, pp.
  2669--2678.

\bibitem{Beckenbach2019}
L.~Beckenbach, P.~Osinenko, and S.~Streif, ``Model predictive control with
  stage cost shaping inspired by reinforcement learning,'' in \emph{Conference
  on Decision and Control (CDC)}, 2019, pp. 7110--7115.

\bibitem{Beckenbach2018}
L.~Beckenbach, P.~Osinenko, T.~G{\"o}hrt, and S.~Streif, ``Constrained and
  stabilizing stacked adaptive dynamic programming and a comparison with model
  predictive control,'' in \emph{European Control Conference (ECC)}, 2018, pp.
  1349--1354.

\bibitem{Beckenbach2020}
L.~Beckenbach, P.~Osinenko, and S.~Streif, ``A {Q}-learning predictive control
  scheme with guaranteed stability,'' \emph{European Journal of Control},
  vol.~56, pp. 167--178, 2020.

\bibitem{Zanon2019}
M.~Zanon, S.~Gros, and A.~Bemporad, ``Practical reinforcement learning of
  stabilizing economic mpc,'' in \emph{2019 18th European Control Conference
  (ECC)}, 2019, pp. 2258--2263.

\bibitem{Zanon2020}
M.~Zanon and S.~Gros, ``Safe reinforcement learning using robust {MPC},''
  \emph{{IEEE} Transactions on Automatic Control}, vol.~66, no.~8, pp.
  3638--3652, 2020.

\bibitem{Berkenkamp2019}
F.~Berkenkamp, ``Safe exploration in reinforcement learning: Theory and
  applications in robotics,'' Ph.D. dissertation, ETH Zurich, 2019.

\bibitem{Perkins2001}
T.~J. Perkins and A.~G. Barto, ``Lyapunov-constrained action sets for
  reinforcement learning,'' in \emph{ICML}, vol.~1, 2001, pp. 409--416.

\bibitem{Chow2018}
Y.~Chow, O.~Nachum, E.~Duenez-Guzman, and M.~Ghavamzadeh, ``A lyapunov-based
  approach to safe reinforcement learning,'' in \emph{Advances in Neural
  Information Processing Systems}, S.~Bengio, H.~Wallach, H.~Larochelle,
  K.~Grauman, N.~Cesa-Bianchi, and R.~Garnett, Eds., vol.~31.\hskip 1em plus
  0.5em minus 0.4em\relax Curran Associates, Inc., 2018.

\bibitem{Osinenko2020a}
P.~Osinenko, L.~Beckenbach, T.~G{\"o}hrt, and S.~Streif, ``A reinforcement
  learning method with real-time closed-loop stability guarantee,''
  \emph{IFAC-PapersOnLine}, vol.~53, no.~2, pp. 8043--8048, 2020, presented at
  IFAC World Congress.

\bibitem{osinenko2022reinforcement}
P.~Osinenko, D.~Dobriborsci, and W.~Aumer, ``Reinforcement learning with
  guarantees: a review,'' \emph{IFAC-PapersOnLine}, vol.~55, no.~15, pp.
  123--128, 2022.

\bibitem{Vamvoudakis2015}
K.~G. Vamvoudakis, M.~F. Miranda, and J.~P. Hespanha, ``Asymptotically stable
  adaptive--optimal control algorithm with saturating actuators and relaxed
  persistence of excitation,'' \emph{IEEE transactions on neural networks and
  learning systems}, vol.~27, no.~11, pp. 2386--2398, 2015.

\bibitem{Osinenko2021i}
P.~Osinenko, G.~Yaremenko, and I.~Osokin, ``A note on stabilizing reinforcement
  learning,'' 2021.

\bibitem{Perkins2002}
T.~J. Perkins and A.~Barto, ``Lyapunov design for safe reinforcement
  learning,'' \emph{J. Mach. Learn. Res.}, vol.~3, pp. 803--832, 2002.

\bibitem{Hewing2020}
L.~Hewing, K.~P. Wabersich, and M.~N. Zeilinger, ``Recursively feasible
  stochastic model predictive control using indirect feedback,''
  \emph{Automatica}, vol. 119, p. 109095, 2020.

\bibitem{Goehrt2019b}
T.~G{\"o}hrt, P.~Osinenko, and S.~Streif, ``Adaptive actor-critic structure for
  parametrized controllers,'' pp. 652--657, 2019, presented at IFAC Symposium
  on Nonlinear Control Systems (NOLCOS).

\bibitem{Goehrt2019}
------, ``Adaptive dynamic programming using {L}yapunov function constraints,''
  \emph{IEEE Control Systems Letters}, vol.~3, no.~4, pp. 901--906, 2019,
  presented at Conference on Decision and Control (CDC).

\bibitem{Goehrt2020}
T.~G{\"o}hrt, F.~Griesing-Scheiwe, P.~Osinenko, and S.~Streif, ``A
  reinforcement learning method with closed-loop stability guarantee for
  systems with unknown parameters,'' \emph{IFAC-PapersOnLine}, vol.~53, no.~2,
  pp. 8157--8162, 2020, presented at IFAC WC.

\bibitem{Clarke1997}
F.~Clarke, Y.~Ledyaev, E.~Sontag, and A.~Subbotin, ``Asymptotic controllability
  implies feedback stabilization,'' \emph{IEEE Tran. on Autom. Control},
  vol.~42, no.~10, pp. 1394--1407, 1997.

\bibitem{Osinenko2018}
P.~Osinenko, G.~Devadze, and S.~Streif, ``Analysis of the caratheodory's
  theorem on dynamical system trajectories under numerical uncertainty,''
  \emph{IEEE/CAA Journal of Automatica Sinica}, vol.~5, no.~4, pp. 787--793,
  2018.

\bibitem{Schmidt2021}
P.~Schmidt, P.~Osinenko, and S.~Streif, ``On inf-convolution-based robust
  practical stabilization under computational uncertainty,'' \emph{IEEE
  Transactions on Automatic Control}, 2021.

\bibitem{Osinenko2021d}
P.~Osinenko and G.~Yaremenko, ``On stochastic stabilization of sampled
  systems,'' in \emph{Conference on Decision and Control (CDC)}, 2021,
  accepted.

\bibitem{Osinenko2021e}
P.~Osinenko, G.~Yaremenko, I.~Osokoin, and G.~Malaniia, ``On stochastic
  stabilization via non-smooth control {L}yapunov functions,'' \emph{IEEE
  Transactions on Automatic Control}, 2022, accepted.

\bibitem{Khasminskii2011}
R.~Khasminskii and G.~Milstein, \emph{Stochastic Stability of Differential
  Equations}, ser. Stochastic Modelling and Applied Probability.\hskip 1em plus
  0.5em minus 0.4em\relax Springer, 2011.

\bibitem{Berkenkamp2017}
F.~Berkenkamp, M.~Turchetta, A.~Schoellig, and A.~Krause, ``Safe model-based
  reinforcement learning with stability guarantees,'' in \emph{Advances in
  Neural Information Processing Systems}, I.~Guyon, U.~V. Luxburg, S.~Bengio,
  H.~Wallach, R.~Fergus, S.~Vishwanathan, and R.~Garnett, Eds., vol.~30.\hskip
  1em plus 0.5em minus 0.4em\relax Curran Associates, Inc., 2017.

\bibitem{Kamalapurkar2016}
R.~Kamalapurkar, P.~Walters, and W.~E. Dixon, ``Model-based reinforcement
  learning for approximate optimal regulation,'' \emph{Automatica}, vol.~64,
  pp. 94--104, 2016.

\bibitem{Kamalapurkar2018}
R.~Kamalapurkar, P.~Walters, J.~Rosenfeld, and W.~E. Dixon, \emph{Reinforcement
  learning for optimal feedback control: A Lyapunov-based approach}.\hskip 1em
  plus 0.5em minus 0.4em\relax Springer, 2018.

\bibitem{AlTamimi2008}
A.~Al-Tamimi, F.~L. Lewis, and M.~Abu-Khalaf, ``Discrete-time nonlinear hjb
  solution using approximate dynamic programming: Convergence proof,''
  \emph{IEEE Transactions on Systems, Man, and Cybernetics, Part B
  (Cybernetics)}, vol.~38, no.~4, pp. 943--949, 2008.

\bibitem{Heydari2014}
A.~Heydari, ``Revisiting approximate dynamic programming and its convergence,''
  \emph{IEEE transactions on cybernetics}, vol.~44, no.~12, pp. 2733--2743,
  2014.

\bibitem{Liu2013}
D.~Liu and Q.~Wei, ``Policy iteration adaptive dynamic programming algorithm
  for discrete-time nonlinear systems,'' \emph{IEEE Transactions on Neural
  Networks and Learning Systems}, vol.~25, no.~3, pp. 621--634, 2013.

\bibitem{Jiang2015}
Y.~Jiang and Z.-P. Jiang, ``Global adaptive dynamic programming for
  continuous-time nonlinear systems,'' \emph{IEEE Transactions on Automatic
  Control}, vol.~60, no.~11, pp. 2917--2929, 2015.

\bibitem{Hafstein2005}
S.~F. Hafstein, ``A constructive converse lyapunov theorem on asymptotic
  stability for nonlinear autonomous ordinary differential equations,''
  \emph{Dynamical Systems}, vol.~20, no.~3, pp. 281--299, sep 2005.

\bibitem{Clarke2008}
F.~Clarke, Y.~Ledyaev, R.~Stern, and P.~Wolenski, \emph{Nonsmooth Analysis and
  Control Theory}.\hskip 1em plus 0.5em minus 0.4em\relax Springer, 2008, vol.
  178.

\bibitem{Sastry2013}
S.~Sastry, \emph{Nonlinear systems: analysis, stability, and control}.\hskip
  1em plus 0.5em minus 0.4em\relax Springer Science \& Business Media, 2013,
  vol.~10.

\bibitem{Hafstein2004}
S.~F. Hafstein, ``A constructive converse lyapunov theorem on exponential
  stability,'' \emph{Discrete \& Continuous Dynamical Systems}, vol.~10, no.~3,
  p. 657, 2004.

\bibitem{Baier2014}
R.~Baier and S.~Hafstein, ``Numerical computation of control {L}yapunov
  functions in the sense of generalized gradients,'' in \emph{Int. Symposium on
  Mathematical Theory of Networks and Systems}, 2014, pp. 1173--1180.

\bibitem{Khalil1996}
H.~Khalil, \emph{Nonlinear {S}ystems}.\hskip 1em plus 0.5em minus 0.4em\relax
  Prentice-Hall. 2nd edition, 1996.

\bibitem{Kimura2015}
S.~Kimura, H.~Nakamura, and Y.~Yamashita, ``Asymptotic stabilization of
  two-wheeled mobile robot via locally semiconcave generalized homogeneous
  control {L}yapunov function,'' \emph{SICE J. of Control, Measurement, and
  Syst. Integr.}, vol.~8, no.~2, pp. 122--130, 2015.

\bibitem{rcognita2020}
\BIBentryALTinterwordspacing
P.~Osinenko, ``rcognita: a framework for reinforcement learning agent design
  and simulation,'' 2020. [Online]. Available:
  \url{https://github.com/pavel-osinenko/rcognita}
\BIBentrySTDinterwordspacing

\bibitem{Osinenko2018a}
P.~Osinenko, L.~Beckenbach, and S.~Streif, ``Practical sample-and-hold
  stabilization of nonlinear systems under approximate optimizers,'' \emph{IEEE
  Control Systems Letters}, vol.~2, no.~4, pp. 569--574, 2018, presented at
  Conference on Decision and Control (CDC).

\bibitem{Clarke2011}
F.~Clarke, ``{L}yapunov functions and discontinuous stabilizing feedback,''
  \emph{Annual Reviews in Control}, vol.~35, no.~1, pp. 13--33, 2011.

\bibitem{Domingo2020}
D.~Domingo, A.~d’Onofrio, and F.~Flandoli, ``Properties of bounded stochastic
  processes employed in biophysics,'' \emph{Stochastic Analysis and Appl.},
  vol.~38, no.~2, pp. 277--306, 2020.

\end{thebibliography}

%%%%%%%%%%%%%%%%%%%%%%%%%%%%%%%%%%%%%%%%%%%%%%%%%%%%%%%%%%%%%%%%%%%%%%%%%%%%%%%%%%%%%%%%%%%%
%%%%%%%%%%%%%%%%%%%%%%%%%%%%%%%%%%%%%%%%%%%%%%%%%%%%%%% BIOGRAPHY

\pagebreak

%%%%%%%%%%%%%%%%%%%%%%%%%%%%%%%%%%%%%%%%%%%%%%%%%%%%%%%%%%%%%%%%%%%%%%%%%%%%%%%%%%%%%%%%%%%%
%%%%%%%%%%%%%%%%%%%%%%%%%%%%%%%%%%%%%%%%%%%%%%%%%%%%%%%%%%%%%%%%%%%%%%%%%%%%%%%%%%%%%%%%%%%%
\appendix

%%%%%%%%%%%%%%%%%%%%%%%%%%%%%%%%%%%%%%%%%%%%%%%%%%%%%%%%%%%%%%%%%%%%%%%%%%%%%%%%%%%%%%%%%%%%
%%%%%%%%%%%%%%%%%%%%%%%%%%%%%%%%%%%%%%%%%%%%%%%%%%%%%%%%%%%%%%%%%%%%%%%%%%%%%%%%%%%%%%%%%%%%
\subsection{Preliminaries}
\label{sub:prf-preliminaries}

First note that \eqref{eqn:sys} is in fact equivalent to 
\begin{equation}
\begin{aligned}
	&\diff X_t = f(X_t, U_t)\diff t + \sigma(X_t, U_t)Z'_t \diff t,\\
	&\mathbb{P}\left[X_0 = \text{const}\right] = 1,
\end{aligned}
\end{equation}
where $Z'$ is a measurable bounded random process, such that $\forall t \ \nrm{Z'_t} \leq \lip{Z}$. The latter is motivated by the fact that due to absolute continuity of $Z$, there always exists such $Z'$ that  $Z_t = Z_0 + \int_0^tZ'_\tau\diff \tau$, with the last identity implying $\diff Z_t = Z'_t \diff t$.

Recall that the statements of the presented theorems concern an arbitrary Lyapunov function $L$. If a function $L$ is already known to be a Lyapunov function, then the same applies to $L' := \zeta^{-1} L$, which makes $\zeta L = L'$. Thus without loss of generality let's assume $\zeta = 1$ and denote $w^{\#} := w^{\#}_{1}$.

Once the system's state $X_{k\delta}$ is in $\B_{s^{\ast}} \subset \B_{\nrm{X_0}}$, which is referred to as the \textit{core ball}, the setting of $(U_k,w_k)$ is arbitrary. 
This is dictated by the fact that the optimization problem may become infeasible in a small vicinity of the origin due to noise and S\&H behavior.
Nor is one interested in what happens there as far as S\&H-setting is concerned.

Let $v^\ast = \hatalphalow(r)$. 
At this point, note that for any $w \in \W$,
\begin{align*}
\hatalphalow(\nrm{x}) \leq \hat{J}^w(x) \leq v^\ast \; \implies \; \nrm{x} \leq r
\end{align*}
and also 
\begin{align*}
\hatalphaup(s^{\ast}) \leq \hat{J}^w(x) \leq \hatalphaup(\nrm{x}) \; \implies \; \nrm{x} \geq s^{\ast},
\end{align*}
which relate the value of $\hat{J}$ to the facts that $x \in \B_{r}$ or $s \not \in \B_{s^{\ast}}$, respectively. 
It can be seen that, among other factors to be detailed later, the bounding functions $\hat{\alpha}_{\text{low},\text{up}}$ contribute to the radius of the target ball.

Finally, call an actor-critic sequence ${(U_{k\delta},w_k)}_{k \in \N_0}$ \textit{admissible} if, for any $k \in \N_0$, \eqref{eqn:actor-critic-stab-1}--\eqref{eqn:actor-critic-stab-4} are satisfied along the corresponding trajectory as long as $X_{k\delta} \not \in \ball_{s^{\ast}}$.

A single element of an actor-critic sequence will be called an \emph{actor-critic tuple}.

First, let $S^\ast > \sublyapunov$  and  let $\lip{f}$ be the Lipschitz constant of $f$ on $\B_{S^\ast}$. 

Now, let

\begin{equation}
\label{eqn:bound-drift}
\bar{f} := \sup_{\subalign{\hspace{4pt}x &\in \B_{S^\ast} \\ u &\in \mathbb{U} }} \; \nrm{f(x, u)}.
\end{equation}

\begin{lem}
\label{lem:escape-1}
If for some $t$ it is true that $\nrm{X_t} \leq \sublyapunov$, then $X_{t + \hat{t}} \in \ball_{S^\ast}$ is implied by $\hat{t} < \frac{S^\ast - \sublyapunov}{\bar{f} + \sigma_{\max} \lip{Z}}$.
\end{lem}
\begin{proof}
Let's assume that by $t + \hat{t}$, $\nrm{X_{t + \hat{t}}}$ has reached $S^\ast$ for the first time. Then, obviously
$$
S^\ast - \sublyapunov \leq \nrm{X_{t + \hat{t}}} - \nrm{X_t} \leq (\bar{f} + \sigma_{\max} \lip{Z})\hat{t}. 
$$
This in turn means that $\hat{t}$ would have to be at least as large as $\frac{S^\ast - \sublyapunov}{\bar{f} + \sigma_{\max} \lip{Z}}$ for the agent to be able to transcend $S^\ast$ by that time.

\end{proof}

%% !!! Note that we do not actually account for what happens after s^{\ast} is reached. In theory it can lead to \nrm{X} escaping B_{S^\ast}, thus invalidating the bounds. Therefore all of our derivations are based upon the assumption that s^{\ast} has not yet been reached. (nor will it be reached in the course of the current step)
\begin{lem}
\label{lem:bound-transition}
Let the following constraints be in place:
\begin{multline}
G^{\bar \nu}_\eps(w_k, w_{k-1}, \Lambda_{U_{k\delta}}(X_{k\delta}), X_{k\delta}) \leq 0,\\
G^{\bar \nu}_\eps(w_{k + 1}, w_{k}, \Lambda_{U_{(k + 1)\delta}}(X_{(k + 1)\delta}), X_{(k + 1)\delta}) \leq 0,\\
\end{multline}
then if $\delta$ is sufficiently small and $\eps < \bar \nu - (1 + \frac{1}{3})\lip{\hat{J}}\sigma_{\max}\lip{Z}$, the following implication holds:
\begin{equation}
\begin{aligned}
\hat{J}^{w_{k}}(X_{k\delta}) \leq \alpha_{\text{up}}(\nrm{X_{0}})
\implies\\
\hat{J}^{w_{k + 1}}(X_{(k + 1)\delta}) \leq \alpha_{\text{up}}(\nrm{X_{0}})
\end{aligned}
\end{equation}
\end{lem}
\begin{proof}
First, observe that $\hat{J}^{w_{k}}(X_{k\delta}) \leq \alpha_{\text{up}}(\nrm{X_{0}})$ implies
\begin{multline}
 \hat{\alpha}_{\text{low}}(\nrm{X_{k\delta}}) \leq \hat{J}^{w_{k}}(X_{k\delta}) \leq \alpha_{\text{up}}(\nrm{X_{0}}),\\
 \nrm{X_{k\delta}} \leq \hat{\alpha}^{-1}_{\text{low}}(\alpha_{\text{up}}(\nrm{X_{0}})).
\end{multline}
Together with \textit{Lemma \ref{lem:escape-1}} the above yields
\begin{equation}
\label{eqn:no-escape}
\forall t\in[k\delta, (k + 1)\delta] \ X_{t} \in \ball_{S^\ast}.
\end{equation}
The latter justifies the usage of earlier mentioned bounds and Lipschitz constants $\bar f, \lip{\hat{J}}, \lip{f}$ that assume $X_t \in \ball_{S^\ast}$. This enables us to derive
\begin{multline}
\nrm{X_{(k + 1)\delta} - \Lambda_{U_{k\delta}}(X_{k\delta})} \leq
\\ \leq \, \nrm{ \int_{k \delta}^{(k+1) \delta} f({X_t},U_{k\delta}) \, \text{d}t - f(X_{k\delta},U_{k\delta}) \delta  }  + \\  \hspace{0.8em} 
\nrm{ \int_{k \delta}^{(k+1) \delta} \sigma(X_t, U_{k\delta})Z'_t \, \text{d}t}\\
\leq \, \nrm{  \int_{k \delta}^{(k+1) \delta} f({X_t},U_{k\delta})  - f(X_{k\delta},U_{k\delta}) \, \text{d}t  } + \\  \hspace{0.8em} 
\nrm{ \int_{k \delta}^{(k+1) \delta} \sigma(X_t, U_{k\delta})Z'_t \, \text{d}t} \\
\leq \, \int_{k \delta}^{(k+1) \delta} \nrm{f({X_t},U_{k\delta})  - f(X_{k\delta},U_{k\delta})} \, \text{d}t  + \\  \hspace{0.8em} 
\int_{k \delta}^{(k+1) \delta} \nrm{\sigma(X_t, U_{k\delta})Z'_t} \, \text{d}t \\ 
\leq \,  \int_{k \delta}^{(k+1) \delta} \lip{f} (\bar{f} + \sigma_{\max}\lip{Z}) \delta \, \text{d}t  + \\  \hspace{0.8em} 
\int_{k \delta}^{(k+1) \delta} \sigma_{\max}\lip{Z} \, \text{d}t \\
= \,  \underbrace{\lip{f} (\bar{f} + \sigma_{\max}\lip{Z}) \delta^2 + \sigma_{\max}\lip{Z}\delta}\limits_{\chi_1(\delta) := }.
\end{multline}
Due to $G^{\bar \nu}_\eps(\dots) \leq 0$ the latter implies

\begin{multline}
\label{eqn:critic-bound-1}
\hat{J}^{w_{k + 1}}(X_{(k + 1)\delta}) \leq \hat{J}^{w_{k}}(X_{(k + 1)\delta})  + \varepsilon_{1} \leq \\ 
\hat{J}^{w_{k}}(\Lambda_{U_{k\delta}}(X_{k\delta})) + \lip{\hat{J}}\chi_{1}(\delta) + \varepsilon_{1} \leq \\
\hat{J}^{w_{k}}(X_{k\delta}) - \delta{\bar \nu} + \eps +  \lip{\hat{J}}\chi_{1}(\delta) \leq  \\
\alpha_{\text{up}}(\nrm{X_{0}}) - \delta{\bar \nu} + \eps  + \lip{\hat{J}}\chi_{1}(\delta) \leq \\
\alpha_{\text{up}}(\nrm{X_{0}}) + \eps +  \delta(\lip{\hat{J}}\sigma_{\max}\lip{Z} - {\bar \nu}) + \\ \delta^2\lip{\hat{J}}\lip{f}(\bar f + \sigma_{\max}\lip{Z}).
\end{multline}
Now, pick $\delta$ to be sufficiently small:
\begin{multline}
\label{eqn:sufficiently-small-1}
\delta < \frac{{\bar \nu}}{2\lip{\hat{J}}\lip{f}(\bar f + \sigma_{\max}\lip{Z})}.\\
\end{multline}
Substituting \eqref{eqn:sufficiently-small-1} into \eqref{eqn:critic-bound-1} yields
\begin{multline}
\hat{J}^{w_{k + 1}}(X_{(k + 1)\delta}) \leq \alpha_{\text{up}}(\nrm{X_{0}}).\\
\end{multline}
\end{proof}

\begin{lem}
Let $K \in \mathbb{N}$ and $w_0 = w^\#$. If ${\bar \nu} > 2\lip{\hat{J}}\sigma_{\max}\lip{Z}+2\frac{\eps}{\delta}$, then for a sufficiently small $\delta$ the following implication holds:
\begin{equation}
\label{eqn:no-escape-2}
\begin{aligned}
&\forall k \leq K \ \ G^{\bar \nu}_\eps(w_k, w_{k-1}, \Lambda_{U_{k\delta}}(X_{k\delta}), X_{k\delta}) \leq 0 \implies\\
&\forall t \in [0, (K + 1)\delta] \ X_t \in \ball_{S^\ast}.
\end{aligned}
\end{equation}
\end{lem}
\begin{proof}
Note, that
\begin{multline}
\label{eqn:base-1}
\hat{J}^{w_{0}}(X_{0}) = L(X_{0}) \leq \alpha_{\text{up}}(\nrm{X_{0}}).\\
\end{multline}
By performing mathematical induction with \eqref{eqn:base-1} as its base and \textit{Lemma \ref{lem:bound-transition}} as its step we obtain:
\begin{multline}
\forall k \leq K \ \hat{J}^{w_{k}}(X_{k\delta}) \leq \alpha_{\text{up}}(\nrm{X_{0}}).\\
\end{multline}
Together with \eqref{eqn:no-escape} the above yields \eqref{eqn:no-escape-2}.
\end{proof}

The latter lemma justifies the choice of Lipschitz constants and bounds such as for instance \eqref{eqn:bound-drift}. The said bounds will have been made heavy use of in the proofs of the presented theorems.

Next, notice that the right-hand side of \eqref{eqn:sys} is continuous in $x$ in each $k$th sampling step.
Noticing that both $f$ and $\sigma$ are Lipschitz continuous in the first argument furnishes the existence of local (path-wise) solutions.
Combining this with the fact that $S^\ast$ is not exceeded we conclude existence of global solutions.

%%%%%%%%%%%%%%%%%%%%%%%%%%%%%%%%%%%%%%%%%%%%%%%%%%%%%%%%%%%%%%%%%%%%%%%%%%%%%%%%%%%%%%%%%%%%
%%%%%%%%%%%%%%%%%%%%%%%%%%%%%%%%%%%%%%%%%%%%%%%%%%%%%%%%%%%%%%%%%%%%%%%%%%%%%%%%%%%%%%%%%%%%
\subsection{Proof of Theorem \ref{thm:stab}}
\label{sub:prf-stability}

Let the following hold:
\begin{multline}
	\forall k \in \mathbb{N}_{0}  \ G^{{\bar \nu}}_\eps(w_k, w_{k-1}, \Lambda_{U_{k\delta}}(X_{k\delta}), X_{k\delta}) \leq 0 \\ \lor \ X_{k\delta} \in \ball_{s^{\ast}}. \\
\end{multline}

If one assumes \eqref{eqn:sufficiently-small-1} together with $\eps < \bar \nu - (1 + \frac{1}{3})\lip{\hat{J}}\sigma_{\max}\lip{Z}$, then for $X_{k\delta} \in \ball_{S^{\ast}}\setminus\ball_{s^{\ast}}$ we get
\begin{align} \label{subeq:decay-Jhat-min}
\begin{split}
& \exists \Delta < 0 \ \hat{J}^{w_{k+1}}(X_{(k+1)\delta}) - \hat{J}^{w_k}(X_{k\delta}) \leq \\
& \leq \; -\delta \bar{\nu} + \lip{\hat{J}} \lip{f} (\bar{f} + \sigma_{\max}\lip{Z} ) \delta^2 + \\ &\hspace{1.2em} +  \lip{\hat{J}}\sigma_{\max}\lip{Z}\delta + \eps < \Delta.
\end{split}
\end{align}

Now if we assume the opposite of \eqref{eqn:implication-in-thm1} that leads to
\begin{multline}
\forall k \in \mathbb{N}_{0} \ \hat{J}^{w_{k}}(X_{k\delta}) < \hat{J}^{w_0}(X_{0}) + k\Delta,\\
\implies \exists k \in \mathbb{N}_{0} \ \hat{J}^{w_{k}}(X_{k\delta}) < 0.
\end{multline}
which is in turn a contradiction, since $J^w$ is non-negative.

Thus, it can be concluded that with a sufficiently small $\delta > 0$, the state of the system will be driven to $\ball_{s^{\ast}}$. 
The time it takes to reach $\ball_{s^{\ast}}$ can be determined in a uniform way from the decay rate and the value of $\hat{J}$ \cite{Clarke1997,Osinenko2020a,Schmidt2021,Osinenko2018a}. 

%%%%%%%%%%%%%%%%%%%%%%%%%%%%%%%%%%%%%%%%%%%%%%%%%%%%%%%%%%%%%%%%%%%%%%%%%%%%%%%%%%%%%%%%%%%%
%%%%%%%%%%%%%%%%%%%%%%%%%%%%%%%%%%%%%%%%%%%%%%%%%%%%%%%%%%%%%%%%%%%%%%%%%%%%%%%%%%%%%%%%%%%%
\subsection{Proof of Theorem \ref{thm:feas}}
\label{sub:prf-feasibility}

Let  $\bar \nu := 4 \lip{J}\sigma_{\max}\lip{Z}$.

Consider a current state $X_{k\delta} \in \B_{S^\ast} \setminus \B_{s^{\ast}}$ at some time step $k \in \N_0$.
If $\delta \leq \bar \delta$ and $U_{k\delta} := \eta(X_{k\delta})$, then by \textit{Assumption} \ref{asm:CLF}-\ref{asm:sample-wise-decay}, it holds that
\begin{enumerate}
	\item[i)] $\hat{J}^{w^{\#}}(X_{(k + 1)\delta}) - \hat{J}^{w^{\#}}(X_{k\delta}) \leq - \delta \nu(X_{k\delta})$,
	\item[ii)] $ \alphalow(\|X_{k\delta}\|) \leq \hat{J}^{w^{\#}}(X_{k\delta}) \leq \alphaup(\|X_{k\delta}\|)$.
\end{enumerate}
Given these properties, it needs to be shown that \eqref{eqn:actor-critic-stab-1}--\eqref{eqn:actor-critic-stab-4} are well posed to mean that these constraints are feasible for all times where $X_{k\delta} \not\in \B_{s^{\ast}}$. 

First let's ensure that $\delta > 0$ is sufficiently small and select and adequate $\eps$.
\begin{multline}
\delta < \frac{\sigma_{\max}\lip{Z}}{3\lip{f}(\bar f + \sigma_{\max}\lip{Z})},\\
\eps := (2 + \frac{2}{3})\lip{J}\sigma_{\max}\lip{Z} = \frac{\bar \nu}{2} + \frac{\bar \nu}{6}.
\end{multline}
Note, that the latter implies
\begin{multline}
 \lip{\hat{J}} \lip{f} (\bar{f} + \sigma_{\max}\lip{Z})\delta < \frac{\lip{\hat{J}}\sigma_{\max}\lip{Z}}{3} < \frac{\bar \nu}{12}.\\
\end{multline}

From i), it follows that

\begin{align*}
&\hat{J}^{w^{\#}}(\Lambda_{U_{k\delta}}(X_{k\delta})) - \hat{J}^{w^{\#}}(X_{k\delta}) \\ &\leq - \delta \nu(X_{k\delta}) +\\ &\hspace{1.2em}   \lip{\hat{J}} \lip{f} (\bar{f} + \sigma_{\max}\lip{Z} ) \delta^2 + \lip{\hat{J}}\sigma_{\max}\lip{Z}\delta\\
&\leq \, - \delta (\nu(X_{k\delta}) - \frac{\bar \nu}{12} - \frac{\bar \nu}{4}) \leq -\delta \bar \nu.
\end{align*}

Hence, (C3) is satisfied.

%	\red{Leave some interval between the bounds}
Now, note that
\begin{multline}
\label{eqn:JeqL}
\hat J^{w^\#}(\cdot) := L(\cdot), \\
\end{multline}
which directly implies that (C2) is satisfied.

Also, \eqref{eqn:JeqL} implies
\begin{multline}
\alpha_{\text{low}}(\nrm{X_{k\delta}}) \leq \hat J^{w^{\#}}(X_{k\delta}) \leq \alpha_{\text{up}}(\nrm{X_{k\delta}}) \implies\\
\hat \alpha_{\text{low}}(\nrm{X_{k\delta}}) \leq \hat J^{w^{\#}}(X_{k\delta}) \leq \hat \alpha_{\text{up}}(\nrm{X_{k\delta}}).\\ 
\end{multline}
The latter indicates that (C4) is satisfied.

Observe that
\begin{multline}
\label{eqn:c1-backward}
J^{w^{k - 1}}(X_{k\delta}) \geq J^{w^{k - 1}}(\Lambda_{U_{(k - 1)\delta}}(X_{(k - 1)\delta})) - \\ 
\lip{J}\nrm{X_{k\delta} -  \Lambda_{U_{(k - 1)\delta}}(X_{(k - 1)\delta})} \geq \\
 L(\Lambda_{U_{(k - 1)\delta}}(X_{(k - 1)\delta})) -  
 \lip{\hat{J}} \lip{f} (\bar{f} + \sigma_{\max}\lip{Z} ) \delta^2 - \\ \lip{\hat{J}}\sigma_{\max}\lip{Z}\delta,
\end{multline}
while at the same time
\begin{multline}
\label{eqn:c1-premise}
\hat J^{w^{\#}}(X_{k\delta}) - L(\Lambda_{U_{(k - 1)\delta}}(X_{(k - 1)\delta})) + \\ \lip{\hat{J}} \lip{f} (\bar{f} + \sigma_{\max}\lip{Z} ) \delta^2 + \lip{\hat{J}}\sigma_{\max}\lip{Z}\delta \leq \\ 2\lip{\hat{J}} \lip{f} (\bar{f} + \sigma_{\max}\lip{Z} ) \delta^2 + 2\lip{\hat{J}}\sigma_{\max}\lip{Z}\delta \leq \\ \frac{\bar \nu}{2}\delta + \frac{\bar \nu}{6}\delta = \eps\delta.
\end{multline}
Substituting \eqref{eqn:c1-premise} into \eqref{eqn:c1-backward} yields
\begin{multline}
J^{w^{\#}}(X_{k\delta}) - J^{w^{k - 1}}(X_{k\delta}) \leq \eps\delta,\\
\end{multline}
which is equivalent to (C1). 
 
Hence, feasibility of CALF is shown.

%%%%%%%%%%%%%%%%%%%%%%%%%%%%%%%%%%%%%%%%%%%%%%%%%%%%%%%%%%%%%%%%%%%%%%%%%%%%%%%%%%%%%%%%%%%%
%%%%%%%%%%%%%%%%%%%%%%%%%%%%%%%%%%%%%%%%%%%%%%%%%%%%%%%%%%%%%%%%%%%%%%%%%%%%%%%%%%%%%%%%%%%%
\subsection{Proof of Theorem \ref{thm:CALF-2}}
\label{sub:CALF-2}
%Set  $\bar \nu := 4 \lip{J}\sigma_{\max}\lip{Z}$.

The proof of the theorem technically follows that of Theorem \ref{thm:stab}.
The difference is in invocations of the stabilizing policy.

Consider Algorithm \ref{alg:no-lyapunov}.
Let us denote, for brevity:
\begin{equation}
	\label{eqn:CALF2_notation1}
	\begin{aligned}
		& x^\circ := x_{\text{prev}}, \\
		& \hat J^\circ := \hat J^{w_{\text{prev}}}(x^\circ). \\
	\end{aligned}
\end{equation}

Also, let us denote:
\begin{equation}
	\label{eqn:CALF2_notation2}
	\begin{aligned}
		& \hat{\mathbb K} := \{k \in \N_0 : \text{critic finds a solution } w^* \}, \\
		& \mathbb K_L := \{k \in \N_0 : \eta \text{ is invoked} \}, \\
	\end{aligned}
\end{equation}
the sets of time step indices where the critic succeeds to find a solution $w^*$ and, respectively, where it does not and so the stabilizing policy $\eta$ is invoked.

Now, define:
\begin{equation*}
%	\label{eqn:CALF2_notation1}
	\begin{aligned}
		& x^\circ_k := 	\begin{cases}
							x_k, k \in \hat{\mathbb K},\\
							x^\circ_{k-1}, k \in \mathbb K_L,
						\end{cases}
	\end{aligned}
\end{equation*}
and, by the same token, $\hat J^\circ_k$.

Notice that, although the user of Algorithm \ref{alg:no-lyapunov} might not be aware of a Lyapunov function associated with a known stabilizing policy, such a Lyapunov function always exists \cite{Clarke2011}.
By Assumption \ref{asm:sample-wise-decay}, $L$ is a Lyapunov function associated with $\eta$ and $\nu$ is the corresponding sample-wise decay rate.
Notice realizations of CALF as per Algorithms \ref{alg:CALF_AC}, \ref{alg:CALF_SARSA} require knowledge of such an $L$, whereas Algorithm \ref{alg:no-lyapunov} only assumes its existence which is justified by converse theorems \cite{Clarke2011}.

Let $L_k$ denote $L(X_{k \delta})$.
When running Algorithm \ref{alg:no-lyapunov}, we expect the behavior of $L_k$ and $\hat J^\circ_k$ to look schematically like Figure \ref{fig_CALF2_effLF_schematic} depicts.

\begin{figure}[h]
\centering
\includegraphics[width=0.9\columnwidth]{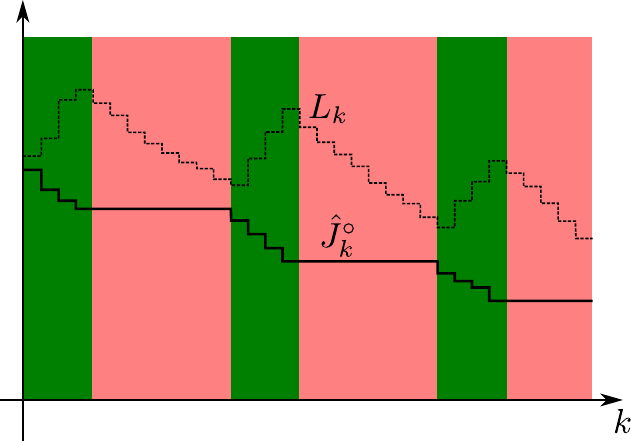}
\caption{Schematic of the critic and a Lyapunov function behavior under Algorithm \ref{alg:no-lyapunov}.}
\label{fig_CALF2_effLF_schematic}
\end{figure}

%Let's assume that the system's state never reaches $\ball_{s^{\ast}}$.

It holds that
\[
	\forall k \in \N_0 \spc \nrm{x^\circ_k} \le \hatalphalow\inv(\hat J^\circ_k) \le \hatalphalow\inv(\hat J^\circ_0).
\]

Observe that
\[
	\begin{aligned}
	& \forall k \in \N_0 \\
	& \qquad L(x^\circ_k) \le \alphaup(\nrm{x^\circ_k}) \le \alphaup( \hatalphalow\inv(\hat J^\circ_k) ) \le \alphaup( \hatalphalow\inv(\hat J^\circ_0) ).
	\end{aligned}
\]

Since for all $k \in \mathbb K_L$, $L_k$ is non-increasing, we conclude that
\[
	\begin{aligned}
	& \forall k \in \N_0 \spc L_k \le \alphaup( \hatalphalow\inv(\hat J^\circ_0) ),
	\end{aligned}
\]
whence
\[
	\begin{aligned}
	& \forall k \in \N_0 \spc \nrm{X_{k\delta}} \le \alphalow\inv ( \alphaup( \hatalphalow\inv(\hat J^\circ_0) ) ),
	\end{aligned}
\]

%First notice that whenever at some time step $k\delta$ a successful critic update occurs, we have:
%$\nrm{X_{k\delta}} \leq \hat \alpha_{\text{low}}^{-1}(J^{w_{k}}(X_{k\delta})) \leq \hat \alpha_{\text{low}}^{-1}(J^{w_{0}}(X_{0}))$. 
%This bound can thus only be exceeded, when the nominal policy is invoked, 
Thus, an overshoot bound $S^{\ast}$ may be set so as to satisfy:
\begin{multline}
S^{\ast} > \alphalow\inv ( \alphaup( \hatalphalow\inv(\hat J^\circ_0) ) ).\\
\end{multline}

Now, if $X_{k\delta} \notin \ball_{s^{\ast}}$, then
\begin{align}
	\nu(X_{k\delta}) > 6 \lip{J}\sigma_{\max}\lip{Z}.
\end{align}

We have, denoting $\Delta \bullet_k \triangleq \bullet_{k+1} - \bullet_k$:
\begin{equation}
	\label{eqn:CALF2_decays}
	\begin{aligned}
		& \forall k \in \hat{\mathbb K} \spc \Delta \hat J^\circ_k \le -6 \lip{J}\sigma_{\max}\lip{Z}\delta, \\
		& \forall k \in \mathbb K_L \spc \Delta L_k \le -6 \lip{J}\sigma_{\max}\lip{Z}\delta. \\		
	\end{aligned}
\end{equation}

If the critic always succeeded, $\ball_{s^{\ast}}$ would be reached after no more than
\begin{equation}
	\label{eqn:CALF2_critic_reaching}
		\hat T := \frac{\hat J^\circ_0 - \hatalphalow(s^*)}{6 \lip{J}\sigma_{\max}\lip{Z}\delta}
\end{equation}
steps.
If the stabilizing policy were always invoked instead, $\ball_{s^{\ast}}$ would be reached after no more than
\begin{equation*}
%	\label{eqn:CALF2_LF_reaching}
		T_L := \frac{L_0 - \alphalow(s^*)}{6 \lip{J}\sigma_{\max}\lip{Z}\delta}
\end{equation*}
steps.

But, unlike $\hat J^\circ_0$, the Lyapunov function $L_k$ may step-wise increase when $k \in \hat{\mathbb K}$ (see Figure \ref{fig_CALF2_effLF_schematic}).
However, $L_k$ never exceeds $\alphaup( \hatalphalow\inv(\hat J^\circ_0) )$ as shown above.
Thus, let us define 
\begin{equation}
	\label{eqn:CALF2_LF_reaching}
		T_L := \frac{\alphaup( \hatalphalow\inv(\hat J^\circ_0) ) - \alphalow(s^*)}{6 \lip{J}\sigma_{\max}\lip{Z}\delta}.
\end{equation}
This is a reaching time of $\ball_{s^{\ast}}$ if only $\eta$ were invoked after a worst growth of $L_k$ after a successful critic update.
We say ``a reaching time" to stress that it is effectively just a bound.
Now, let $T^* := \max\{\hat T, T_L\}$.

We now argue what can happen till the ball $\ball_{s^{\ast}}$ is reached.
The two limiting cases, when only the critic succeeds or $\eta$ is involved, are clear from the reaching time $T^*$.
That is, $\ball_{s^{\ast}}$ is reached within $T^*$ steps.
We need thus to argue about the mixed case.
It is easy to see that the worst case is when per every critic success, there is $T^*-1$ invocations of $\eta$ until $\ball_{s^{\ast}}$ is ``almost" reached, but not quite, followed by another critic success and so on.
Thus, the overall reaching time of $\ball_{s^{\ast}}$ under Algorithm \ref{alg:no-lyapunov} is
\begin{equation}
	\label{eqn:CALF2_reaching}
	T^*(T^*-1).
\end{equation}

%Now, let's consider two alternative cases:
%\begin{enumerate}
%\item  There was an infinite number of successful critic updates.
%\item There was a finite number of successful critic updates.
%\end{enumerate} 

%If 1) holds, then at some point the value of $\hat J$ will turn negative, which violates the constraints implied by 1), thus we have a contradiction.

%Now, let's assume that 2) holds. That means that from some point onwards ($N$) there will be an infinite number of consecutive applications of the stabilizing policy.

%Since $X_{k\delta} \notin \ball_{s^{\ast}}$, we have
%\begin{multline}
%\nu(X_{k\delta}) > 6 \lip{J}\sigma_{\max}\lip{Z}.\\
%\end{multline}

%Thus we obtain:
%\begin{multline}
%\forall k \geq N \ L(X_{(k + 1)\delta}) - L(X_{k\delta}) \leq -6 \lip{J}\sigma_{\max}\lip{Z}\delta. \\
%\end{multline}
%The latter trivially implies that $L$ eventually turns negative.
%
%Thus the fact that \textit{Algorithm \ref{alg:no-lyapunov}} eventually drives the state to $\ball_{s^{\ast}}$ is proved by contradiction.
%
%The above naturally yields the following bound on reaching time:
%\begin{multline}
%T_{\text{reaching}} \leq \Biggl\lceil \frac{L(X_{0}) - \alpha_{1}(s^{\ast})}{6 \lip{J}\sigma_{\max}\lip{Z}\delta} \Biggr\rceil.\\
%\end{multline}

\begin{rem}
Since the number of invocations of $\eta$ is not greater than $T^*$ till the $\ball_{s^{\ast}}$ is reached, the critic $\hat J^\circ$ is effectively a multi-step Lyapunov function \ie $\hat J^\circ$ is non-increasing and
\[
	\forall k \in \N_0 \spc \hat J^\circ_{k+T^*} - \hat J^\circ_k < 0.
\]
\end{rem}

\begin{rem}
One may ask why not to always use Algorithm \ref{alg:no-lyapunov} if it lifts the requirement on the knowledge of a Lyapunov function.
The answer is that the reaching time of Algorithm \ref{alg:no-lyapunov} is much more conservative than that of realizations of CALF that do use a Lyapunov function like Algorithm \ref{alg:CALF_AC}.
In fact, there are classes of systems for which Lyapunov functions can be effectively constructed.
For instance, the physical energy can be sometimes be used for a Lyapunov function candidate.
Furthermore, there is systematic procedures like backstepping which help extend a Lyapunov function to higher dimensions etc.
Therefore, it is recommended to use a CALF realization like Algorithm \ref{alg:CALF_AC} when a Lyapunov function is known, and a realization like Algorithm \ref{alg:no-lyapunov} otherwise.
\end{rem}

%%%%%%%%%%%%%%%%%%%%%%%%%%%%%%%%%%%%%%%%%%%%%%%%%%%%%%%%%%%%%%%%%%%%%%%%%%%%%%%%%%%%%%%%%%%%
%%%%%%%%%%%%%%%%%%%%%%%%%%%%%%%%%%%%%%%%%%%%%%%%%%%%%%%%%%%%%%%%%%%%%%%%%%%%%%%%%%%%%%%%%%%%
\subsection{Extra simulation study}
\label{sub:extra_sim}

This section provides graphical material on some of the results obtained with Algorithm \ref{alg:no-lyapunov} obtained using \texttt{rcognita}.

\begin{figure*}[h]
\centering
\includegraphics[width=0.9\textwidth]{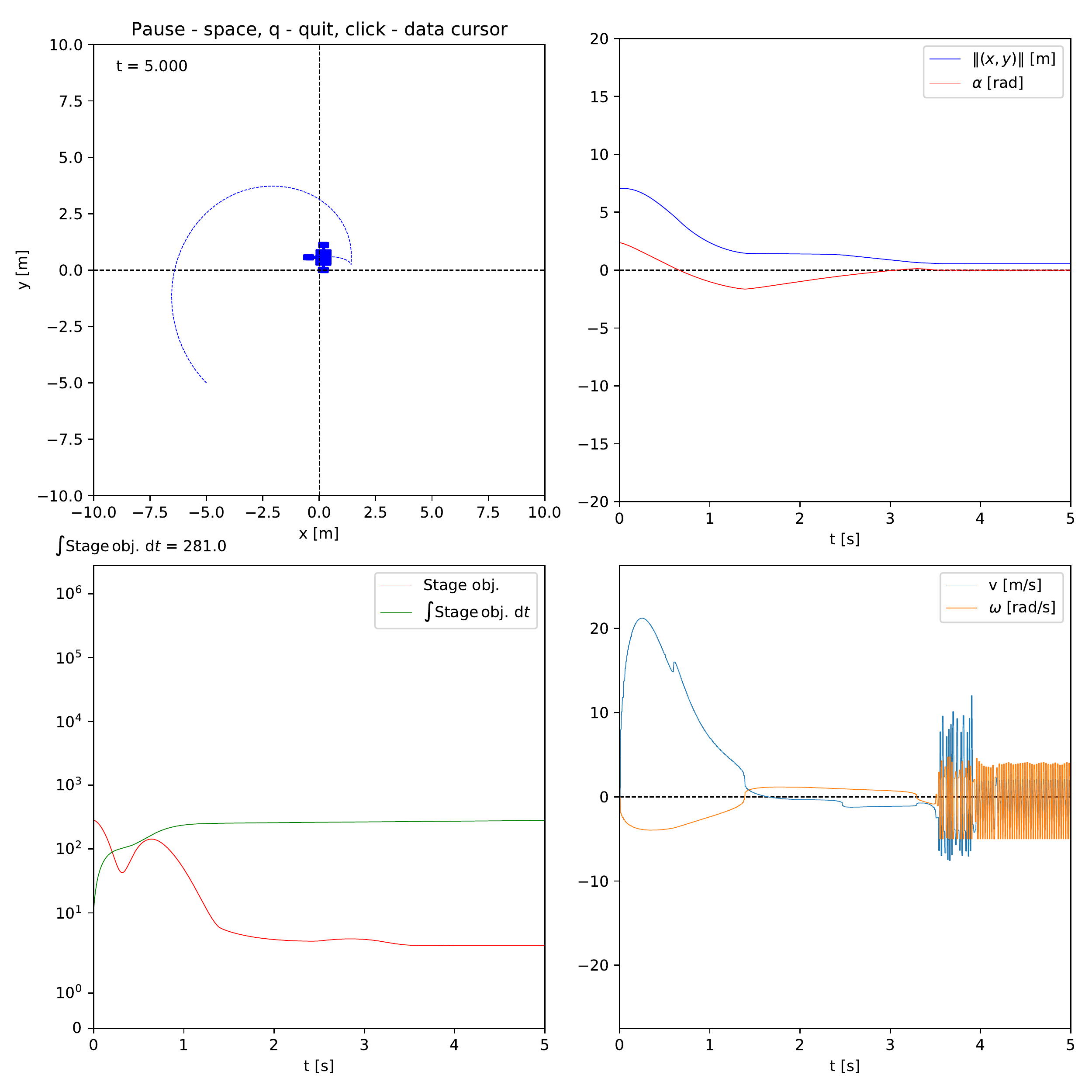}
\caption{Robot parking using the nominal stabilizing policy.
The dither in control near the origin is caused by controller sampling.
The accumulated cost is given just for a reference -- the stabilizing controller is per se not concerned with any optimization.}
\label{fig_nominal}
\end{figure*}

\begin{figure*}[h]
\centering
\includegraphics[width=0.9\textwidth]{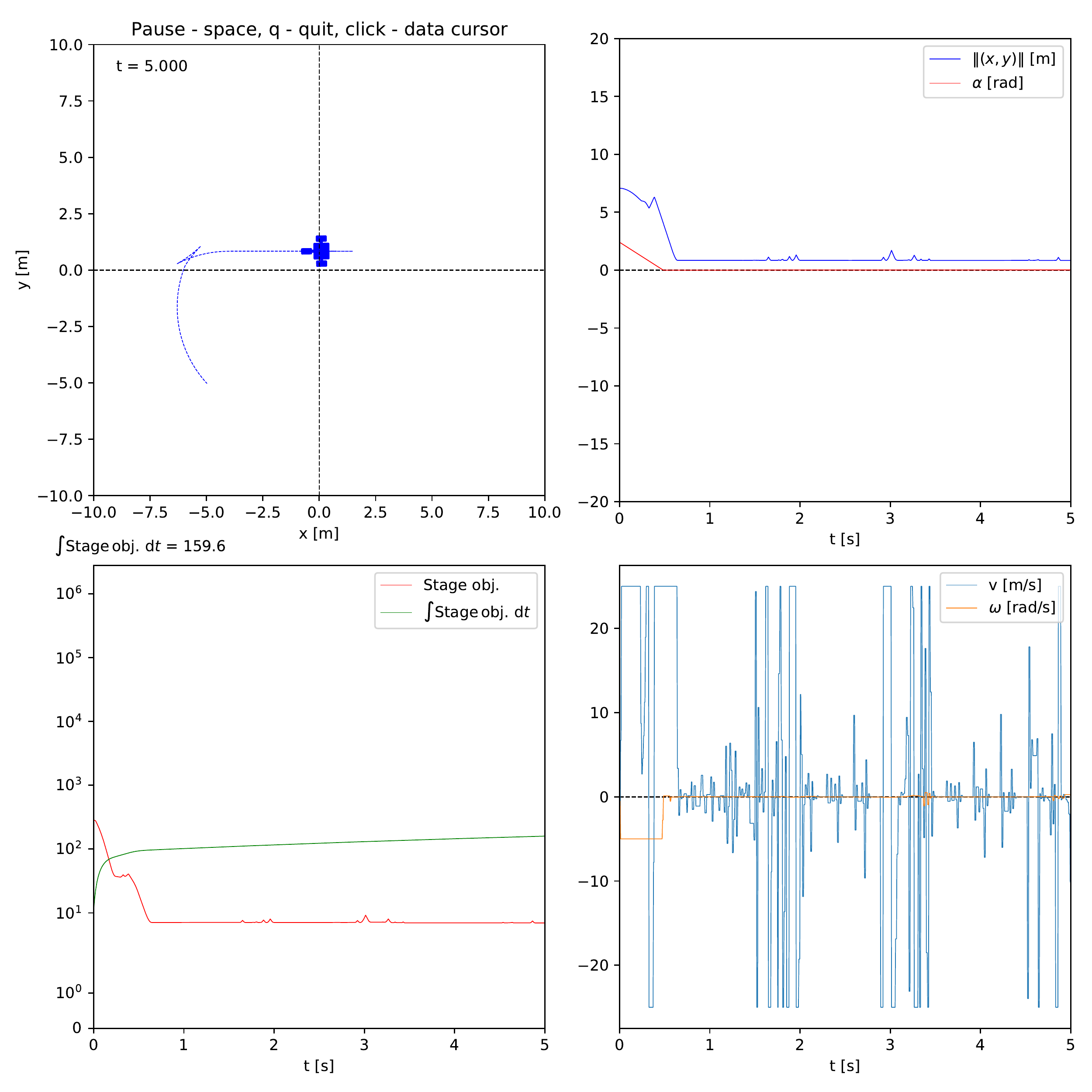}
\caption{Robot parking using the bench-marking reinforcement learning agent without stabilizing constraints.
The parking precision is worse than with the stabilizing controller.
}
%Notice the effect of penalization of the $y^c$-coordinate.}
\label{fig_plainrl}
\end{figure*}

\begin{figure*}[h]
\centering
\includegraphics[width=0.9\textwidth]{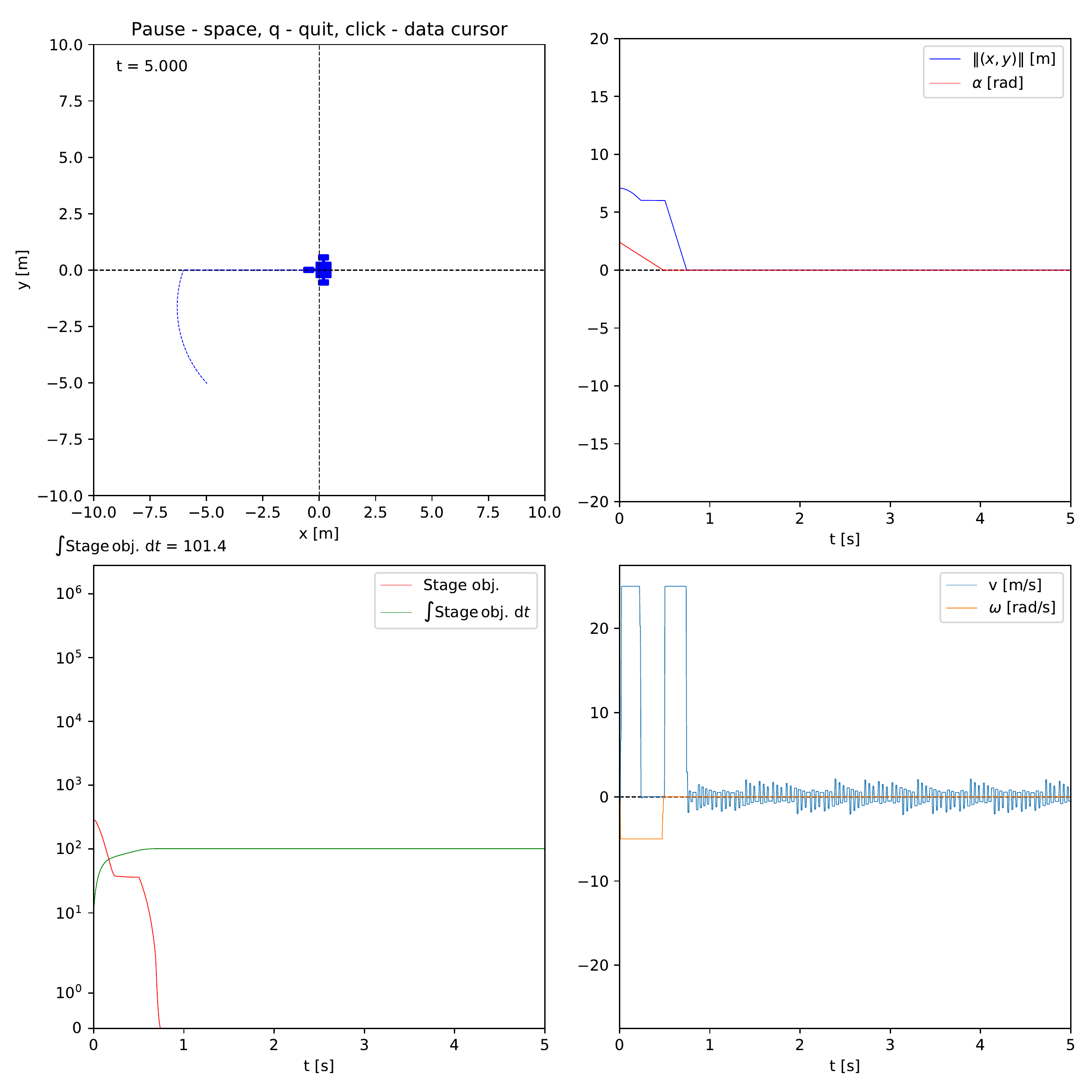}
\caption{Robot parking using CALF as in Algorithm \ref{alg:no-lyapunov}.
The combination of the stabilizing policy with learning allowed for a significantly better stabilization than in both other cases.
The accumulated cost is much lower than with the bench-marking agent.}
\label{fig_stag}
\end{figure*}

%%%%%%%%%%%%%%%%%%%%%%%%%%%%%%%%%%%%%%%%%%%%%%%%%%%%%%%%%%%%%%%%%%%%%%%%%%%%%%%%%%%%%%%%%%%%
%%%%%%%%%%%%%%%%%%%%%%%%%%%%%%%%%%%%%%%%%%%%%%%%%%%%%%%%%%%%%%%%%%%%%%%%%%%%%%%%%%%%%%%%%%%%
\subsection{Bounded noise models}
\label{sub:noise-models}

%\begin{equation}
%	\label{eqn:sys}
%	\pdiff{}{t} X_t = f(X_t, U_t) + \sigma(X_t, U_t)W_t,
%\end{equation}

This section briefly reflects the key models from the survey \cite{Domingo2020}.
Some Latin and Greek letters here are to be interpreted separately from the rest of the text and thus not to be confused.
The easiest way to achieve bounded noise is to apply a saturation function the standard Brownian motion $B_t$.
Such is the case of the sine-Wiener process $Z'_t = \sin \left ( \sqrt{\tfrac {2} {\tau_a}} Z'_t\right)$ with an autocorrelation time parameter $\tau_a$.
%Substituting $Z'_t$ in the system SDE \eqref{eqn:sys} for such a noise would require extra attention to the existence and uniqueness of strong solutions, though.
Another way is to augment the system description with a dynamical noise model.
Thus, the overall description would read, for instance:
\begin{equation}
	\label{eqn:sys-bounded-noise}
%	\tag{sys-aug}
	\begin{aligned}
		& \diff X_t = (f(X_t, U_t) + \sigma(X_t, U_t)Z'_t) \diff t, \\
		& \diff Z'_t = \zeta(Z'_t) \diff t + \chi(Z'_t) \diff B_t,	
	\end{aligned}
\end{equation}
where $\{Z'_t\}$ is the noise process with an internal model described by the drift function $\zeta$ and diffusion function $\Lambda$.
Particular ways to construct the respective stochastic differential equations include the following \cite{Domingo2020}:
\begin{itemize}
%	\item The Sine-Wiener noise realization:	
%		\begin{equation}
%			\label{eqn:SW}
%			\tag{SW}
%			Z'_t = \sin \left ( \sqrt{\tfrac 2 \tau} Z'_t\right), \spc \tau > 0
%		\end{equation} 
%		where $\tau$ is the characteristic autocorrelation time of the process.		
	\item The Doering-Cai-Lin (DCL) noise \\
%	Next equation provides a unique strong solution bounded in $[-1, 1]$.
	\begin{equation}
	\label{eqn:DCL}
%		\tag{DCL}
		\diff Z'_t = - \tfrac 1 w Z'_t \diff t + \sqrt{\tfrac {1 - {Z'_t}^2 }{b_1 (b_2 + 1 )}} \diff B_t,
	\end{equation}
	with parameters $b_2>-1, b_1>0$;
	\item The Tsallis-Stariolo-Borland (TSB) noise \\
%	The diffusion of the process is constant, while the drift tends to infinity in this case while boundaries $[-1,1]$ are approached:
	\begin{equation}
		\label{eqn:TSB}
%		\tag{TSB}
		\diff Z'_t = - \tfrac 1 w \tfrac{Z'_t}{1 - {Z'_t}^2} \diff t + \sqrt{\tfrac{1 - b_2}{b_1}} \diff B_t,
	\end{equation}
	with $b_1 > 0, \spc b_2 < 1$ parameters;
	\item Kessler–S{\o}rensen (KS) noise \\
%	Let $Z'_t$ be the \eqref{eqn:DCL} noise where $\delta \leq 0$ \\
	\begin{equation}
		\label{eqn:KS}
%		\tag{KS}
		\diff Z'_t = - \tfrac {b_3} {\pi b_1} \tan \left ( {\tfrac \pi 2 Z'_t} \right ) \diff t + \tfrac {2} {\pi \sqrt{b_1 (b_2 + 1 )}}  \diff B_t,
	\end{equation}   
	with $b_1 > 0, b_2 \ge 0, b_3 = \tfrac{2b_2+1}{b_2+1}$ parameters.	
\end{itemize}
The above models essentially design the drift and/or diffusion functions so as to confine the strong solutions to stay within $(-1,1)$ (component-wise) almost surely (the unitary bound is chosen for simplicity and may be adjusted according to the application).
It should be noted that the corresponding functions $\zeta, \Lambda$ do not satisfy Lipschitz conditions in the standard way.
Nevertheless, existence and uniqueness of strong solutions can be ensured \cite{Domingo2020}.
So, for instance, in the case of the TSB noise, the drift is at least locally Lipschitz.
This fact, together with non-reachability of the boundaries $-1,1$ (which can be shown) furnishes the strong solution existence and uniqueness.

\end{document}